\documentclass[journal]{IEEEtran}

\makeatletter
\def\ps@headings{%
\def\@oddhead{\mbox{}\scriptsize\rightmark \hfil \thepage}%
\def\@evenhead{\scriptsize\thepage \hfil \leftmark\mbox{}}%
\def\@oddfoot{}%
\def\@evenfoot{}}
\makeatother
\usepackage{amssymb}
\usepackage{amsfonts}
\usepackage{latexsym}
\usepackage[mathscr]{euscript}
\usepackage{amsmath}
\usepackage{graphicx}
\usepackage{algorithm}
\usepackage{algorithmic}

\newcommand{\algorithmicfunction}{\textbf{function }}
\newcommand{\algorithmicendfunction}{\textbf{end}}

\newtheorem{theorem}{Theorem}[section]
\newtheorem{example}{Example}[section]

\newtheorem{remark}{Remark}[section]
\newtheorem{corollary}{Corollary}[section]

\usepackage{color} %allows to mark some entries in the tables with color

\makeindex
\begin{document}

%\title{On the Nash Stability in the Hedonic Coalition Formation Games \tnoteref{t1}} 
\title{Green Broadcast Transmission in Cellular Networks: A Game Theoretic Approach\thanks{This work was supported by the ECOSCells project. It was done within the INRIA \& Alcatel Lucent Bell-Labs joint research laboratory.}}
\author{\IEEEauthorblockN{Cengis Hasan$^{\dag}$, Jean-Marie Gorce$^{\dag}$ and Eitan Altman$^{\ddag}$} \\ \small
\IEEEauthorblockA{$^{\dag}$Inria, University of Lyon, INSA-Lyon, 6 Avenue des Arts 69621 Villeurbanne Cedex, France\\
$^{\ddag}$Inria, 2004 Route des Lucioles, 06902 Sophia-Antipolis Cedex, France \\
$\{$cengis.hasan, eitan.altman$\}$@inria.fr, $\{$jean-marie.gorce$\}$@insa-lyon.fr
}
}
%\author[citi]{Cengis Hasan\corref{cor1}}
%\ead{cengis.hasan@inria.fr}
%\author[inria]{Eitan Altman}
%\ead{eitan.altman@inria.fr}
%\author[citi]{Jean-Marie Gorce}
%\ead{jean-marie.gorce@insa-lyon.fr}
%\address[citi]{Inria, University of Lyon, INSA-Lyon, 6 Avenue des Arts 69621 Villeurbanne Cedex, France}
%\address[inria]{Inria, 2004 Route des Lucioles, 06902 Sophia-Antipolis Cedex, France}
%\cortext[cor1]{Corresponding author}

\date{}
\maketitle

%\chapter{The Mobile Assignment Problem in Broadcast Transmission: Optimization Aspects}\label{chapter:TheMobileAssignmentProbleminBroadcastTransmissionOptimization Aspects}

\begin{abstract}
This paper addresses the mobile assignment problem in a multi-cell broadcast transmission seeking minimal total power consumption by considering both transmission and operational powers. While the large scale nature of the problem entails to find distributed solutions, game theory appears to be a natural tool. We propose a novel distributed algorithm based on group formation games, called \textit{the hedonic decision algorithm}. This formalism is constructive: a new class of group formation games is introduced where the utility of players within a group is separable and symmetric being a generalized version of parity-affiliation games. The proposed hedonic decision algorithm is also suitable for any set-covering problem. To evaluate the performance of our algorithm, we propose other approaches to which our algorithm is compared. We first develop a centralized recursive algorithm called \textit{the hold minimum} being able to find the optimal assignments. However, because of the NP-hard complexity of the mobile assignment problem, we propose a centralized polynomial-time heuristic algorithm called \textit{the column control} producing near-optimal solutions when the operational power costs of base stations are taken into account. Starting from this efficient centralized approach, a \textit{distributed column control algorithm} is also proposed and compared to \textit{the hedonic decision algorithm}. We also implement the nearest base station algorithm which is very simple and intuitive and efficiently manage fast-moving users served by macro BSs. Extensive simulation results are provided and highlight the relative performance of these algorithms. The simulated scenarios are done according to Poisson point processes for both mobiles and base stations.
\end{abstract}

\section{Introduction}
Broadcast scenarios have been widely studied for video or audio broadcasting. More recently, the Multimedia Broadcast/Multicast Service (MBMS) \cite{LTEbook} became a requirement of the Long-term Evolution (LTE) specifications to support the delivery of broadcast/multicast data in LTE systems. Broadcast and multicast downlink transmissions make no significant difference at the physical layer. Basically, broadcast services are available to all users without the need of subscribing to a particular service. Therefore, multicasting can thus be seen as ``broadcast via subscription'', with the possibility of charging for the subscription \cite{LTEbook}.
MBMS is intended to be used for some content, such  as streaming transmission of a sport or cultural event, but broadcasting may also be of interest to transmit some signalling such as a beacon for time synchronization or for power control purposes.

We consider broadcasting under a green-aware objective aiming at reducing the energy consumption which is an important issue in wireless environments\cite{surveyofgreennetworking}. Broadcasting may bring a strong improvement in wireless channels since a common resource (in frequency and/or time) may be used for all destinations. The transmission cost for a base station (BS) to reach all nodes in a multicast group is assumed to be proportional to the power needed to reach the worst mobile
among the group, where the worst refers to the mobile receiving the weaker signal which relies on its distance and on additional shadowing effects. We thus consider the situation where there is one common information that every mobile $m \in M$ is interested to receive, and which can be obtained from any one of $n$ BSs.  The objective is then to achieve a mobiles assignment which \emph{minimizes the total power consumption}. 

As mentioned above, this problem is relevant not only for streaming data transmission but also for signalling. A corollary question is indeed how many base stations should be kept active to ensure the signalling to a group of mobiles even when they are not transmitting. In low activity periods, such as during the night, it could be relevant to keep active only few BS ensuring the coverage of few active mobiles. 
Indeed, energy consumption and electromagnetic pollution are main societal and economical challenges that developed countries have to handle. %A typical \textit{small-cell}'s (femto, pico, or micro cell environment) coverage area is possible to range from a few meters to hundreds of meters. 
%The cell size reduction offers theoretically higher capacity and energy efficiency.  
The evolution of cellular networks toward smaller cells offering theoretically higher capacity could in turn lead to an unacceptable increase of the energy expenditure of wireless systems. When decreasing the cell size, the energy consumed for data transmission becomes lower compared to the \emph{operational power costs} (e.g. power amplifiers, cooler, etc.) of a typical BS. Switching off a BS may then bring significant improvements in energy efficiency. Therefore, we take into account the switching on/off operation in the problem formulation. The overarching problem studied in the sequel is then finding \textit{energy-efficient broadcast transmission techniques} to reduce spurious energy using distributed schemes. 

The mobile assignment problem (MAP) in the context of broadcast transmission that we study in this work is actually a special case of the \textit{simple plant location problem} (SPLP) \cite{SPLPsurvey}. SPLP lies within \textit{clustering problems}. In the MAP, basically, the objective is to assign the points to at most $k$ clusters so that the sum of all distances between points in the same cluster ($k$-clustering) is minimized. In \cite{bilo}, the typical cost for a BS-mobile pair is assumed to be only a function of distance between the BS and the mobile, formulated as $\sum_{N_j\in \mathcal{C}} \max_{i\in N_j} d_{ij}^{\alpha} + P_0^j$. Here, $N_j$ is a \emph{cluster} of mobiles assigned to BS $j$, $\mathcal{C}$ is the set of clusters, $d_{ij}$ is the distance between mobile $i$ and BS $j$, $\alpha$ is the path loss exponent and $P_0^{j}$ is the operational power cost loaded to BS $j$. This formulation is modified in order to consider the effect of shadowing leading to the following total cost $\sum_{N_j\in \mathcal{C}} \max_{i\in N_j} d_{ij}^{\alpha}/\Psi_{ij} + P_0^j$ where $\Psi_{ij}$ denotes the shadowing effect between mobile $i$ and BS $j$. Thus, this modification turns the MAP into the SPLP. 
%We utilize the advantage of broadcast transmission in problem formulation in order to decrease the complexity.

While finding the global minimum of the MAP may be identified as an NP-hard problem from SPLP literature, the large scale nature of the cellular network further requires to solve it in a decentralized manner. Thus, game theory appears as a natural tool to cope with both features: distributed decision and NP-hardness. We address this problem by considering the mobiles as players being able to make strategic decisions and the BSs as the strategy identifiers: each mobile has to choose the best BS to be served. 

\subsection{Related Work}
Computational geometric approaches to the MAP can be found in \cite{bilo,
Lev-TovandPeleg, mincostcoverage, coveringpoints}. In \cite{Lev-TovandPeleg},
the authors examined the 1-dimensional version of the MAP, where the effects of
shadowing and operational power cost are not taken into account. Polynomial
time solutions via dynamic programming are proposed. In \cite{mincostcoverage}, authors suggested approximation algorithms (and an algebraic intractability result) for selecting an optimal line on which to place BSs to cover mobiles, and a proof of NP-hardness for any path loss exponent $\alpha > 1$.

The papers \cite{RodopluMeng,wieselthier,egecioglu,CagaljHubaux} focused on source-initiated broadcasting of data in static all-wireless networks. Data are distributed from a source node to each node in a network. The main objective is to construct a minimum-energy broadcast tree rooted at the source node. Multi-hop routing is not the scope of our paper. 

In \cite{hasanaltmangorce}, the combined problem of (i) deciding what
subset of the mobiles would be assigned to each BS, and then (ii) sharing the
BSs' cost of multicast among the mobiles is studied. The subset that is wished to assign to a given BS is said to be its target set of mobiles. This problem can be conceived as a coalitional pricing game played by mobiles which is called \emph{the association game of mobiles}.

\subsection{Our Contribution}
We propose algorithmic solutions for mobile assignment in the context of broadcasting, in order to minimize the overall energy consumption related to transmission and operational powers.
In this context, switching off some fraction of BSs is considered to be a way of decreasing dramatically the
total energy consumption. Note however that heterogeneous networks include macro and small-cells with or without coordination. It is reasonable to assume that small-cells are subject to switching off operation while macro-cells are always turned on. They can indeed serve moving mobiles in order to decrease the number of hand-offs. Further, since the small-cells are deployed intensively, their transmission power is lower than those of macro-cells, while their circuit power dominates. 

Comparing the transmission power about few milliWatts with the
operational power costs which may approach tens of Watts, turning off a fraction of BSs is appealing for reducing the total energy footprint of the network. The efforts for turning off some BSs will be concentrated on small-cells serving fixed mobiles\cite{thefutureofSCN}.

The referred literature mostly concentrates on the geometric aspects of the MAP where basically, the coverage area of a BS is assumed to be a disc which issues from omnidirectional antenna pattern. However, \emph{the effect of
shadowing}, special designed \emph{antenna patterns} as well as \emph{the
operational power costs} may impact the BS-mobile assignments. In this
paper, we take into account these effects by introducing a \emph{power cost matrix} containing all BS-mobile pairing power costs. 

Furthermore, several papers working on coverage optimization deal with static optimization and planning from a centralized point of view. But the dynamic switching-off process associated to the large-scale nature of the network induces to find distributed solutions. To this end, we deal with this problem through a group formation game formulation.
Subsequently, we introduce a new algorithm based on group formation games, called \textit{hedonic decision algorithm}. This formalism is constructive: a new class of group formation games is introduced where the utility of players within a group is separable and symmetric. This is a generalization of parity affiliation games and this hedonic decision algorithm is in fact applicable for any set covering problem. 

To prove the efficiency of this approach, we then derive four other methods allowing to solve the initial problem. 
First, we propose a recursive algorithm called \emph{the hold minimum
algorithm} which solves the considered problem optimally. However, the hold minimum
algorithm operates in a centralized way since it requires the whole knowledge
for each BS-mobile pairing power cost. We then adapt an approach from the SPLP literature to the MAP: a centralized polynomial-time heuristic algorithm is proposed called the \emph{the column control} which produces optimal assignments when taking into account the operational power cost. This algorithm is also extended to a distributed approach, where each mobile gathers the local information from the BSs located in its range. On the other hand, \emph{the nearest base station algorithm}, a distributed greedy algorithm which runs in polynomial-time is also evaluated. This algorithm is not efficient if the operational power cost is large, but is very efficient for the fast-moving users served by macro BSs. 

The rest of the paper is organized as follows. In section \ref{sec:genericproblem}, the MAP is formulated mathematically as a clustering problem and different formulations are then proposed. In section \ref{sec:decentralized}, the game framework and the hedonic decision algorithm are proposed. In section \ref{sec:efficientalgorithms}, we derive other algorithmic solutions for the MAP and their complexity is anlayzed in section \ref{sec:timecomplexity}. Finally, we present simulation results in section \ref{sec:simulationresults} and we expose some conclusions in section \ref{sec:conclusions}.

\section{The Generic MAP Problem}
\label{sec:genericproblem}
\begin{figure*}
\includegraphics[width=\linewidth]{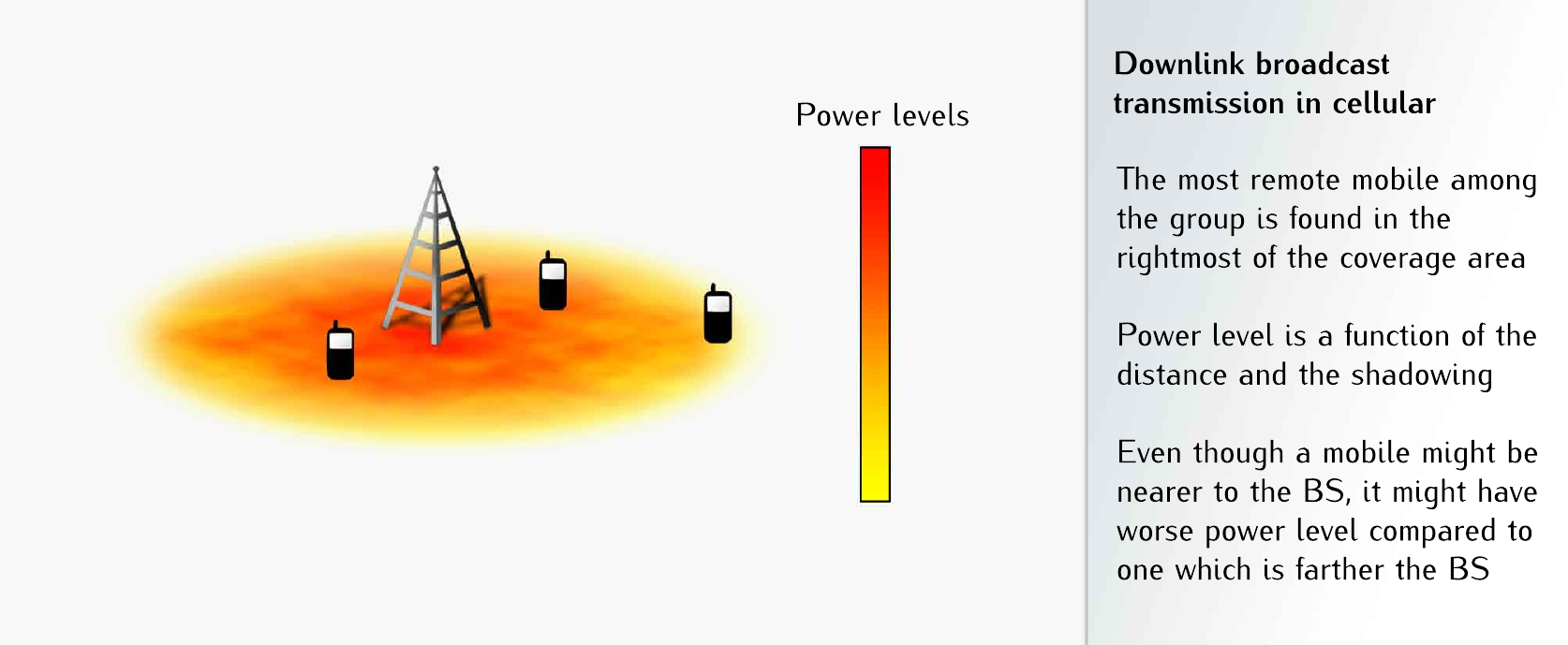}
\caption{Broadcast transmission in cellular networks.}
\label{fig:broadcasttransmission}
\end{figure*}

We consider the coverage problem in the case of broadcast transmission in cellular networks. We assume that each BS transmits simultaneously to the mobiles. The distance between the mobile $i$ and BS $j$ is represented by $d_{ij}$. The power needed to receive the transmission is given by $P_r$. We consider basic signal propagation model capturing path loss as well as shadowing effect formulated as
\begin{equation}\label{propModel}
    P_{ij} = P_r \frac{d_{ij}^{\alpha}}{\Psi_{ij}},
\end{equation}
where $P_{ij}$ and $\alpha$ denote transmitted power from BS $j$ to mobile $i$ and path loss exponent, respectively. The random variable $\Psi$ is used to model slow fading effects and commonly follows a log-normal distribution.%, i.e., the variable $10 \log_{10} \Psi$ follows a normal distribution.

The required transmission power depends on the mobile having the worst
signal level from the BS (Figure \ref{fig:broadcasttransmission}). At this
power level, all mobiles are guaranteed to receive a sufficient power.
We also consider the operational power cost denoted as $P_0^j$
which captures the energy expenditure of a typical BS $j$ for operational
costs (power amplifiers, cooler, etc.). So, the total power cost (transmission
power + operational power cost) of a typical transmission between BS $j$ and
mobile $i$ is denoted as
\begin{equation}
    p_{ij} = P_{ij} + P_0^j.
\end{equation}
Let $M = (1,\ldots,m)$ and $N = (1,\ldots,n)$ be the sets of mobiles and BSs, 
respectively. Representing the \emph{power cost matrix} $\mathbf{P} = 
(p_{ij})\in \Re^{m \times n}$, we assume $p_{ij}\in [0,\infty)$ where 
if $P_{ij} > P_{max}$, then $p_{ij} = \infty$ ($P_{max}$ denotes a maximal 
power, for instance, in WiFi, it is $100$ mW). 

\subsection{The MAP as a Clustering Problem}
Clustering is a rich branch of combinatorial problems which have been extensively studied in many fields including database systems, image processing, data mining, molecular biology, etc. \cite{bilo}. Consider the set of mobiles $M$; a \textit{cluster} is any non-empty subset of $M$ and a \textit{clustering} is a partition of $M$. Many different clustering problems can be defined. The mostly studied problems are defined through their objective which is to assign the points to at most $k$ clusters so that either:
\begin{itemize}
\item \textit{k-centre}: the maximum distance from any point to its cluster centre is minimized,
\item \textit{k-median}: the sum of distances from each point to its closest cluster centre is minimized,
\item \textit{k-clustering}: the sum of all distances between points in the same cluster is minimized.
\end{itemize}

The problem of clustering a set of points into a specific number of clusters so as to minimize the sum of cluster sizes is referred to as min-size $k$-clustering problem. In \cite{bilo}, the typical cost for a BS-mobile pair is assumed to be only a function of the BS-mobile distance and leads to the formulation: $\sum_{N_j\in \mathcal{C}} \max_{i\in N_j} d_{ij}^{\alpha} + P_0^j$, where $N_j$ is a \emph{cluster} of mobiles assigned to BS $j$, $\mathcal{C}$ is the set of clusters, $d_{ij}$ is the distance between mobile $i$ and BS $j$, $\alpha$ is the path loss exponent, $P_0^{j}$ is the operational power cost loaded to BS $j$. Since in this paper, we add the effect of shadowing to such a cost, the cost function becomes $\sum_{N_j\in \mathcal{C}} \max_{i\in N_j} d_{ij}^{\alpha}/\Psi_{ij} + P_0^j$. Note that the shadowing effect breaks the monotonicity with $d_{ij}$ (see Figure \ref{fig:broadcasttransmission}) and then shifts the problem to a \textit{simple plant location problem} (SPLP) \cite{SPLPsurvey} formulated as follows:
\begin{itemize}
\item Let have $ n $ potential facility locations. A facility can be opened in any location $ j $; opening a facility location has a non-negative cost corresponding to $ P_0^j $ in the MAP. Each open facility can provide an unlimited amount of commodity corresponding to unlimited number of mobiles served by a BS in the MAP;
\item there are $ m $ customers that require a service. The goal is to determine a subset of the set of potential facility locations, at which to open facilities and an assignment of all clients to these facilities so as to minimize the overall total cost.
\end{itemize}

However, the SPLP formulation presents a very high complexity, and we rather propose to turn out the problem as a set covering problem, starting from a binary integer formulation.

\subsection{Binary Integer Formulation}
Note that there are at most $ 2^m - 1 $ possible subsets of $ M $.  Each subset can be associated to any BS and the number of combinations is given by $\eta = n (2^m-1)$. We note $\mathcal{S}$ the collection of total possibilities. The index set of $\mathcal{S}$ is denoted by $L = (1, \ldots, \eta)$. Let $\mathbf{W} = (w_{kl}) \in \Re^{(m+n)\times \eta}$ be a 0-1 matrix with $w_{kl} = 1$ if the node $k$ (i.e. mobile or BS) belongs to the set $(S_l;j)$.  Let $q = (q_l) \in \Re^{\eta}$ be an $\eta$-dimensional vector. The value $q_l$ represents the optimal power of $(S_l; j) \in \mathcal{S}$ by which we denote a pair which consists of a set of mobiles $S_l$ assigned to BS $j$. Clearly
\begin{equation}
    q_l = \max_{i\in S_l} p_{ij}.
\end{equation}
The formulation of this problem is given by
\begin{equation}\label{problem}
\begin{split}
    (P) \qquad & p = \min \sum_{l\in L} q_l x_l \\
              & \text{s.t.} \quad \sum_{l \in L} w_{kl} x_{l} = 1, \quad k\in M\cup N, \\
              & x_l \in \{0,1\}, \quad l\in L
\end{split}
\end{equation}
where the term $\sum_{l \in L} w_{kl} x_{l} = 1$ imposes that only one BS is associated to a mobile. By this way, we do not let a mobile to be assigned to several BSs. It follows that the optimal clustering is denoted as $\mathcal{C}^*\subset \mathcal{S}$ such that $\mathcal{C}^* = \{ (S^{*};j^*) \in \mathcal{S}, \forall j\in N \}$ where $(S^{*};j^*)$ is the optimal pairing.

Thanks to this formulation, we can show that the MAP may be derived as a set partitioning problem. In the MAP, the set $M$ is associated with another set $N$. Therefore, the collection $\mathcal{S}$ contains those subsets of $M$ each of which is associated with every element of the set $N$. Consider the following example.

\begin{example}
Let us have a power cost matrix given by
\begin{equation}
    \mathbf{P} = \begin{bmatrix}
                   3 & 6 \\
                   5 & 1 \\
                 \end{bmatrix}.
\end{equation}
The collection of total possibilities:
\begin{equation}
\mathcal{S} = \{(1;1), (2;1), (1,2;1), (1;2), (2;2), (1,2;2)\}.
\end{equation}
Recall that $(S_l;j)$ denotes the cluster of mobiles $S_l$ assigned to BS $j$. The optimal values for each possibility is given by $q = (q_l) =
(3,5,5,6,1,6)$. Then, we define the following matrix:
\begin{equation}
\mathbf{W} = \begin{bmatrix}
                   1 & 0 & 1 & 1 & 0 & 1 \\
                   0 & 1 & 1 & 0 & 1 & 1 \\
                   1 & 1 & 1 & 0 & 0 & 0 \\
                   0 & 0 & 0 & 1 & 1 & 1
                 \end{bmatrix}.
\end{equation}
The optimal total power is thus calculated by the following binary integer
program:
\begin{equation}
\begin{split}
    p = & \min (3x_1+5x_2+5x_3+6x_4+x_5+6x_6) \\
    \text{s.t. } & x_1+x_3+x_4+x_6 = 1, \\ 
    & x_2+x_3+x_5+x_6=1, \\
    & x_1+x_2+x_3 = 1, \\
    & x_4+x_5+x_6=1, \\
              & x_l \in \{0,1\}, \quad l\in (1,\ldots,6).
\end{split}
\end{equation}
The values $x_1=1$ and $x_5=1$ result in the optimal total power of the
example scenario, i.e., $p=3+1=4$ with the optimal clustering $\mathcal{C}^* = \{(1;1),(2;2)\}$.
\end{example}

However, set partitioning problems are well known to be NP-hard \cite{BalasPadbergSetPartitioning}. Consequently, the MAP being a special set partitioning problem is also NP-hard.% The collection set of the MAP possesses a large scale nature. For example, even for $m = 30$, $n = 10$, the size of collection set is $\kappa = 10( 2^{30}-1) \approx 1.074\times 10^{10}$.
In the next section, we propose to reduce this complexity.

\subsection{The MAP as a Set Covering Problem}
The previous formulation stated that a unique BS is allowed to serve a mobile. However in terms of pure coverage considerations, the optimal solution may feature some mobiles to be covered by several BS, no matter to which BS the mobile eventually associate with. We then relax the condition of associating only one BS to a cluster of mobiles in $ (P) $ such that it is now possible to have a cluster of mobiles covered by more than one BS: $ \sum_{l\in L} w_{kl}x_l \geq 1$. Thus, this arrangement turns the MAP into so called \textit{set covering} problem. 
%Actually, this kind of transmission where there are multiple transmitters and only one receiver is called as MISO. Therefore, we assume that the mobiles are capable of reception this kind of transmission.

Consider a set of mobiles $S$ assigned to BS $j$ noted as $(S;j)$. When a group of mobiles $T\subset S$ deviates to another BS $k$ then the cost due to $S$ becomes additive. Let the cost of $(S;j)$ and $(T;k)$ be $q_S = \max_{i\in (S;j)} p_{ij}$ and $q_T = \max_{i\in (T;k)} p_{ik}$, respectively. We denote the total cost before deviation of $T$ as $p$ and after deviation of $T$ as $p'$, respectively \footnote{Note that these costs are not optimal. We only would like to show what is the effect of deviation of a group of mobiles from theirs current BS.}, which can be given by
\begin{equation}
	p = p_r + q_S,
\end{equation}
\begin{equation}
	p' = p_r' + q_{S\backslash T} + q_T,
\end{equation}
where $p_r$ and $p_r'$ are the remaining costs before and after deviation, respectively. There is always a potential (probability) increasing the total cost when a deviation occurs, i.e. $p' \geq p$. For better observation, let us consider the following power cost matrix:
\begin{equation}
\mathbf{P} = \begin{bmatrix}
                   10 & 15 & 25 \\
                   27 & 20 & 33 \\
                   32 & 31 & 30
                 \end{bmatrix}.
\end{equation}
Let $(S;j) = (1,2,3;1)$ and $(T;k) = (1,2;3)$, respectively. Then, $q_S = \max (10,27,32) = 32$, $q_T = \max (25,33) = 33$, and $q_{S\backslash T} = \max (32) = 32$ resulting in the following total costs:
\begin{align}
	p & = 32 \\
	p' & = 32 + 33 = 65,
\end{align}
where $p_r = p_r' = 0$.

Utilizing this property, 
\begin{quote}
\textit{we delete from the collection $\mathcal{S}$ all those assignments $(S\backslash T;j)$ whenever the cost of $(S;j)$ is equal to the cost of $(S\backslash T;j)$ such that $T\subset S$.}
\end{quote}
For example, let us consider the last example where $M=(1,2,3)$ and $N=(1,2,3)$. For $j = 1$, all possible assignments are 
$(1;1)$, $ (2;1) $, $ (3;1) $, $ (1,2;1) $, $ (1,3;1) $, $ (2,3;1) $, $(1,2,3;1)$  
corresponding to the cost vector $q = (10, 27, 32, 27, 32, 32, 32)$. Note that the cost of $(3;1)$, $(1,3;1)$, $(2,3;1)$, $(1,2,3;1)$ are equal each other. Therefore, we remove $(3;1)$, $(1,3;1)$, $(2,3;1)$ from the collection. The reduced collection of assignments becomes as following:
\begin{align}
\mathcal{S}' = \{(1;1),&(1,2;1),(1,2,3;1),(1;2),(1,2;2),\notag\\
&(1,2,3;2),(1;3),(1,3;3),(1,2,3;3)\}.
\end{align}
Note that this property is an extension of the geometric ``coverage range''. If the shadowing was not considered, this approach would reduced to defining the maximal coverage range. But with shadowing the mobile requiring the most power is not always the further one. Thus, the binary integer program of finding the solution of the problem is given by
\begin{align}
    p = & \min (10 x_1+27 x_2+32 x_3+15 x_4+20 x_5+31 x_6 + \notag \\
    & 25 x_7 + 30 x_8 + 33 x_9) \notag \\
    \text{s.t.} \quad & x_1+x_2+x_3+x_4+x_5+x_6+x_7+x_8+x_9 \geq 1, \notag \\ 
    & x_2+x_3+x_5+x_6+x_9 \geq 1, \notag \\
    & x_3+x_6+x_8+x_9 \geq 1, \notag \\
              & x_l \in \{0,1\}, \quad l\in (1,\ldots,9).
\end{align}
The solution of this problem is found to be $x_6 = 1$ and $x_l = 0, \forall l \in (1,\ldots,9)$ which fits to the optimal one. 

By such an elimination, the size of the collection of assignments reduces from $n(2^m - 1)$ to $nm$. This prove that the set-covering formulation is much more efficient than the set-partitioning one.

\subsection{Brute-force Search Solution}
Enumerating all possible solutions and choosing the one which produces the lowest cost is known as brute-force search or generate and test. 

We represent by $\mathbf{A} = (a_{ij}) \in \Re^{m \times n}$ the \emph{assignment matrix} where $a_{ij} \in (0,1)$. If mobile $i$ is assigned to BS $j$, then $a_{ij} = 1$, otherwise $a_{ij} = 0$. Notice that each row of the assignment matrix includes only unique ``1'' which means that a mobile is served by only one BS, i.e. $\sum_{j}a_{ij} = 1$. This is not in contradiction with our former remark about the possibility of having a mobile covered by several BS. Here now, we decide to associate a mobile to a BS. So if a mobile is covered by several BSs, the serving BS can be anyone of this covering set. Denoting the collection of the assignment matrices $\mathcal{A}$, actually, we formalize the problem as following:
\begin{equation}
    p = \min_{\mathbf{A}\in\mathcal{A}} \left( \sum_{i\in M} \max_{j\in N} \mathbf{A}\otimes \mathbf{P}\right),
\end{equation}
where $\otimes$ is the element-wise product. Note that the total number of possibilities of assignment matrices can be calculated as $|\mathcal{A}| = n^m$.

\section{Decentralized Solution: The Hedonic Decision (HD) Algorithm}
\label{sec:decentralized}
We now turn to the study of decentralized methods for solving the MAP. Our approach is based on \textit{group formation game} (see \cite{bib:KonishiGroupFormationGame}, \cite{bib:HollardGroupFormationGame}, \cite{bib:MilchtaichGroupFormationGame}).

\subsection{The Game Model}
A group formation game is represented by a triple ${\cal G} = \langle M, N^m, (u_x)_{x\in M}\rangle$ where $M = \{1,2,\ldots,m\}$ is the set of \emph{players} (i.e. the mobiles), $N$ is the set of \emph{strategies} (i.e. the BSs) shared by all the players and $u_x : N^m\rightarrow \Re$ is the \emph{utility function} of player $x \in M$. Each player $x\in M$ chooses exactly one element from the $n$ alternatives in $N$. The choices of players are represented by $\sigma = \{s_1,s_2\ldots, s_m\} \subseteq N^m$ which is called the \emph{strategy-tuple} ($s_x$ shows the strategy chosen by player $x$). A \textit{partition} of players according to strategy-tuple $ \sigma $ is denoted as $ \Pi(\sigma) = \{(G_j^{\sigma})_{j\in N}\}$ where $ G_j^{\sigma} $ is the group of players choosing the strategy $ j $.

We assume the two following conditions and we will see later how we can define the utilities to achieve these conditions:
\begin{enumerate}
\item \textbf{Separability}:
The utility of any player in any group is said to be \textit{separable} if its utility can always be split as a sum according to:
\begin{equation}
u_x(s_x,\sigma_{-x}) = \sum_{y\in G_{s_x}^{\sigma}} v_x(y;s_x),
\end{equation}
where $ v_x : \{M;N\} \rightarrow \Re $ may be interpreted as the gain of player $ x $ from player $ y $ if $ x $ chooses strategy $ s_x $. Note that $ v_x(x;s_x) $ is the utility of player $ x $ when it is the only player choosing strategy $ s_x $. Thus, the separability property states that utility transfers among a group of users sharing the same strategy is done such that the utility granted to one user is a sum of utilities granted individually by each partner in the group.

\item \textbf{Symmetricity}:
The utility is said to be \textit{symmetric} if the individual gain of $ x $ from $ y $ is equal to the gain of $ y $ from $ x $ when they both share the same strategy $ s $:
\begin{equation}
v_x(y;s) = v_y(x;s), \quad \forall x,y \in M.
\end{equation}
Therefore, this symmetric utility can be referred to as: $ v_x(y;s) = v_y(x;s) = v(x,y;s)$, $ \forall x,y\in M $. $v(x,y;s)$ is called the \textit{symmetric bipartite utility} of player $ x $ and $ y $ while the common strategy is $ s $.
\end{enumerate}
Actually, the game defined above is a straightforward generalization of party affiliation games \cite{bib:partyAffiliation}.

\begin{theorem}\label{thm:GisAPotentialGame}
$ {\cal G} $ is a potential game.
\end{theorem}
\begin{proof}
A non-cooperative game is a potential game \cite{bib:PotentialGames} whenever there exists a function $\Phi$
\begin{equation}
\Phi(s_x,\sigma_{-x}) - \Phi(s'_x,\sigma_{-x}) = u_x(s_x,\sigma_{-x}) - u_x(s'_x,\sigma_{-x}),
\end{equation}
meaning that when player $ x $ switches from strategy $s_x$ to $s'_x$ the difference of its utility can be given by the difference of a function $ \Phi $. This function is called a \textit{potential function}. 
\begin{quote}
\textit{The strategy-tuple that maximizes the potential function is a Nash equilibrium in the game.}
\end{quote}

Let us choose as following the potential function $ \Phi $:
\begin{align}
\Phi(\sigma) & = \sum_{j\in N} \left[ \sum_{a\in G_j^{\sigma}} v(a,a;j) + \frac{1}{2}\sum_{a\in G_j^{\sigma}} \sum_{b\in G_j^{\sigma}:b\neq a} v(a,b;j) \right] \notag \\
& = \sum_{j\in N}  \sum_{a\in G_j^{\sigma}} v(a,a;j) + \frac{1}{2} \sum_{j\in N} \sum_{a\in G_j^{\sigma}} \sum_{b\in G_j^{\sigma}:b\neq a} v(a,b;j).
\label{eq:potentialfunction}
\end{align}
Actually, here we take the sum of single player utilities and the half of total symmetric bipartite utilities. Let us rewrite the potential function as following:
\begin{align}
&\Phi(s_x,\sigma_{-x}) =  v(x,x;s_x) + \sum_{y\in G_{s_x}^{\sigma}}v(x,y;s_x) \notag\\ & + \underbrace{\sum_{j\in N}  \sum_{a\in G_j^{\sigma}\setminus x} v(a,a;j)  + \frac{1}{2}\sum_{j\in N} \sum_{a\in G_j^{\sigma}\setminus x} \sum_{b\in G_j^{\sigma}\setminus x:b\neq a} v(a,b;j)}_{I}.
\end{align}
When player $ x $ switches from $ \sigma_x $ to $ \sigma'_x $ (the other players do not change their strategies), then the strategy-tuple is transformed from $ \sigma $ to $ \sigma' $, and the potential becomes
\begin{align}
&\Phi(s'_x,\sigma_{-x}) =  v(x,x;s'_x) + \sum_{y\in G_{s'_x}^{\sigma}}v(x,y;s'_x) \notag\\ & + \underbrace{\sum_{j\in N}  \sum_{a\in G_j^{\sigma'}\setminus x} v(a,a;j)  + \frac{1}{2}\sum_{j\in N} \sum_{a\in G_j^{\sigma'}\setminus x} \sum_{b\in G_j^{\sigma'}\setminus x:b\neq a} v(a,b;j)}_{I'}.
\end{align}
Note that $ I=I' $  since the total utility due to the other players is equal both in $ \sigma $ and $ \sigma' $. Thus, the difference of potentials is given by
\begin{align}
&\Phi(s_x,\sigma_{-x}) - \Phi(s'_x,\sigma_{-x}) =  v(x,x;s_x) - v(x,x;s'_x) \notag \\
& + \sum_{y\in G_{s_x}^{\sigma}}v(x,y;s_x) -  \sum_{y\in G_{s'_x}^{\sigma}}v(x,y;s'_x).
\end{align}
On the other hand, the difference of the utility of player $ x $ is calculated as
\begin{align}
& u_x(s_x,\sigma_{-x}) - u_x(s'_x,\sigma_{-x}) =  v(x,x;s_x) - v(x,x;s'_x) \notag \\
& + \sum_{y\in G_{s_x}^{\sigma}}v(x,y;s_x) -  \sum_{y\in G_{s'_x}^{\sigma}}v(x,y;s'_x).
\end{align}
By this result we conclude that 
\begin{equation}
\Phi(s_x,\sigma_{-x}) - \Phi(s'_x,\sigma_{-x}) = u_x(s_x,\sigma_{-x}) - u_x(s'_x,\sigma_{-x}) 
\end{equation}
which proves that $ {\cal G} $ is a potential game. Thus, $ {\cal G} $ admits a Nash equilibrium in pure strategies $ \sigma^\ast $ which results in the partition $ \Pi(\sigma^\ast) $ and which maximizes the potential function $\Phi$.
\end{proof}

\begin{corollary}
\emph{The proof \ref{thm:GisAPotentialGame} is constructive:} any group formation game possessing separable and symmetric utility gain of players within a group always converges to a pure Nash equilibrium. 
\end{corollary}

%Assume that any pair $(x;j):= g^{xj}$ is a \textit{g-player} where $x\in M\cup \emptyset$. When player $x$ joins player $g^{xj}$, we imagine that player $x$ does not have any effect on player $g^{xj}$; thus, $x\cup g^{xj} := g^{xj}$. Although, when player $y\in M\cup \emptyset$ joins player $g^{xj}$, then there is a new coalition $y\cup g^{xj}$. Moreover, we suppose that any player $g^{xj}$ is forbidden to join any player $g^{yk}$, i.e. there do not exist coalitions among g-players which means that they are passive. 

%This is the hedonic coalition formation condition in a game theoretic aspect. Thus, we are able to formalize as a hedonic game, the decentralized solution of the MAP.

\subsection{Separable and Symmetric Gain Allocation}
Recall that the required power for serving the group of mobiles $G_j^{\sigma}$ by BS $j$ is denoted as $\max_{x\in G_j^{\sigma}} p_{xj}$. We represent by $ U(G_j^{\sigma})  = -\max_{x\in G_j^{\sigma}} p_{xj}\leq 0$, as a \textit{utility} arising due to group $ G_j^{\sigma} $. Note that $ U(G_j^{\sigma}) \leq 0,\forall G_j^{\sigma}\subseteq M, \forall j\in N$ is a monotonically decreasing function \cite{infocom}.

The \textit{clustering profit} due to mobile $ x $ and $ y $ in BS $ j $ is given by 
\begin{equation}
\label{eq:clusteringprofit}
\Delta(x,y;j) = U(x,y;j) - [U(x;j) + U(y;j)]
\end{equation}
where $U(x;j) = -p_{xj}$ is the utility of player $ x $ when is served alone in BS $ j $. Therefore, $ \Delta(x,y;j) = p_{xj}+p_{yj} - \max(p_{xj},p_{yj}) = \min(p_{xj},p_{yj})$.
\begin{remark}
Note that the clustering profit is a useful metric for evaluating a group of mobiles. Whenever a group of mobiles are near each other, then the clustering profit is high; thus, assigning this group to only one BS is almost always efficient. 
\end{remark} 
%\begin{lemma}
%The clustering profit is always superadditive with respect to the increasing population within a group.
%\end{lemma}
%\begin{proof}
%Consider two separate groups $ S $ and $ T $ such that $ S\cap T = \emptyset $. When these groups are served by the same BS, we say that these groups are merged. From \cite{infocom}, we know that $ U(S),\forall S\subseteq N $ is always superadditive: $ U(S) + U(T) \leq U(S\cup T) $, $ \forall S,T\subseteq N $, $ S\cap T = \emptyset $. Then the following is obtained from eq. (\ref{eq:clusteringprofit}): $ \Delta(S) + \sum_{i\in S} U(i;j) + \Delta(T) + \sum_{i\in T} U(i;j) \leq \Delta(S\cup T) + \sum_{i\in S\cup T} U(i;j)$ resulting in
%\begin{equation}
%\Delta(S) + \Delta(T) \leq \Delta(S\cup T), \quad \forall S,T\subseteq N, S\cap T = \emptyset.
%\end{equation}
%\end{proof}

%Assume that the mobiles are \textit{strategic decision makers}. The strategic decision is performed in the following way:
%\begin{quote}
%\textit{mobile $x$ prefers BS $j$ to BS $k$ whenever $\phi_x^{(S_j;j)}>\phi_x^{(S_k;k)}$.}
%\end{quote}

To ensure separability and symmetricity, we propose to choose the symmetric bipartite utility according to:
\begin{equation}
v(x,y;j) =
\begin{cases}
  \theta \Delta(x,y;j), &\mbox{ if } x\neq y  \\
  -p_{xj}, & \mbox{ if } x = y,
\end{cases}
\end{equation}
where $ \theta $ is called as \textit{clustering weight} which is a parameter that must be adjusted according to the environment. We will show later how it can impact the convergence point of the system.  Thus, the utility function of any player $x\in M$ is given by $ u_x(j,\sigma_{-x})  =  \sum_{y\in G_{j}^{\sigma}}v(x,y;j) $ then,
\begin{eqnarray}
\label{eq:HedonicDecisionPreferenceFunction}
\lefteqn{u_x(j,\sigma_{-x}) =} \\ 
&& \begin{cases}
\theta \sum_{y\in G_{j}^{\sigma}} \min (p_{xj},p_{yj}) -p_{xj}, &\mbox{ if } G_{j}^{\sigma} \neq x \notag \\
- p_{xj}, &\mbox{ if } G_{j}^{\sigma} = x.
\end{cases}
\end{eqnarray}

\subsection{Interpretation of Clustering Weight}
Let us rewrite the potential function according to defined symmetric bipartite utility. Then, using eq. \eqref{eq:potentialfunction}, we may express the corresponding potential function which is equal to
\begin{equation}
\Phi(\sigma) = \sum_{j\in N} \left[- \sum_{a\in G_j^{\sigma}}p_{aj} + \frac{\theta}{2} \sum_{a\in G_j^{\sigma}} \sum_{b\in G_j^{\sigma}\backslash a} \min(p_{aj},p_{bj})\right].
\end{equation} 
Now, let us consider that we can order the mobiles associated to BS $ j $, according to the required power. The ordered set of mobiles related to BS $ j $ is represented by $ \tilde{G}_j^{\sigma} $. Observe that we can now compute the potential function as following:
\begin{equation}
\Phi(\sigma) = \sum_{j\in N} \left[- \sum_{a\in \tilde{G}_j^{\sigma}}p_{aj} + \theta \sum_{i=1}^{|\tilde{G}_j^{\sigma}|-1} (|\tilde{G}_j^{\sigma}| - i) p_{ij}\right].
\end{equation} 
According to this result and the description of individual costs, we can state the following properties:
\begin{enumerate}
\item If $ \theta $ is very small, i.e. the dominant term in the potential function and in the symmetric bipartite utility is the individual power $ p_{aj} $, then with a very low $ \theta $, each mobile will privilege an association to the nearest BS. This is exactly the case when $ \theta=0 $.
\item When $ \theta  $ increases a mobile may decide to leave its nearest neighbour if the lost in power is compensated by the second term. Suppose that the mobile wants to associate to a BS $ k $, where all other mobiles currently associated with experience a better channel. Then the gain to associate to this second BS for this user will be $ \theta \sum p_{ij} $, i.e. the sum of all powers of mobiles already associated with this BS, weighted by $ \theta $. It is clear that the mobile will be joining either a cell having already strong power terms or a huge number of users.
\item If $ \theta $ becomes very large, we can expect that all mobiles will converge to the same BS which means that only one BS is active.
\end{enumerate}

\subsection{Network Topology as a Result of Best-reply Dynamics}
Consider the setting in which only one player decides its strategy. It is called as \textit{best-reply dynamics} when a player chooses the strategy which maximizes its utility. When there is no any player which can improve its utility, then this network topology, i.e. $ \Pi(\sigma^\star) $, corresponds to a Nash equilibrium. Note that any local maximum in the potential function is a Nash equilibrium. Therefore, the network topology obtained by best-reply dynamics accounts for a local maximum of the potential function. Total power cost related to $ \Pi(\sigma^\star) $ can be given by
\begin{equation}
p = \sum_{G_j\in \Pi(\sigma^\star) } \max_{i\in G_j} p_{ij}
\end{equation}
where $ G_j $ is the group of mobiles associated with BS $ j $ in the case of stable strategy-tuple $ \sigma^{\star} $.

Assuming that each mobile is capable to discover those BSs that can transmit to it, we can produce a scheduler in the following way: each BS generates a random clock-time for all those mobiles that it can transmit; then each mobile selects randomly a clock-time from those BSs that it can discover. We need to produce the clock-times by such a way that the collision of the turns of mobiles is minimal. In case of a collision, the clock-times of the corresponding mobiles are regenerated by corresponding BSs.

In Algorithm \ref{alg:TheHedonicDecision}, the pseudo-code of the HD is given. Note that this is an algorithm performed in both BS and mobile sides by an exchange of the information in a separated channel.

\begin{algorithm}
\caption{The Hedonic Decision}
\label{alg:TheHedonicDecision}
\begin{algorithmic}
\STATE \textbf{Base Station}:
\STATE Check stability
\WHILE{there is no stability}
\STATE Send information to each mobile about the current partition
\STATE Check stability
\ENDWHILE
\STATE \textbf{Mobile}:
\WHILE{there is no stability}
\STATE Determine the preferred BS according to eq.  \eqref{eq:HedonicDecisionPreferenceFunction}
\STATE Send information to the preferred BS
\ENDWHILE
\end{algorithmic}
\end{algorithm}

\begin{corollary}
In the literature, the use of game models for set covering problems is called as \emph{set covering games} \cite{bib:SelfishSetCovering, bib:CostSharingandStrategyproofMechanismsforSetCoverGames,bib:MechanismDesignForSetCoverGames}. The HD algorithm is a novel approach for set covering games. This algorithm is suitable for any set covering problem and facility activation problems where the agents are allowed to make strategic decisions
\end{corollary}

\section{Efficient Algorithms for the MAP}
\label{sec:efficientalgorithms}
In this section, we propose different algorithmic solutions to evaluate and compare the efficiency of the distributed algorithm. Centralized algorithms exploit the set covering problem formulation since the search space is the smallest. However, binary integer linear programs are known to be NP-complete and thus we introduce two algorithms based on dynamic programming: \emph{the hold minimum algorithm} and \emph{the column control algorithm}. Then, we develop the distributed version of the column control algorithm. A greedy solution of the problem is introduced as \emph{the nearest base station} approach. The nearest BS and the column control are known and already used in the literature for SPLP problems. We adapt these algorithms to the MAP. 

Because of the large scale nature of the collection set $\mathcal{S}$, we rather develop the algorithms by making all operations on the power cost matrix. This approach foster the iterative removals of elements in the collection set, and ensure a faster convergence.

\subsection{Optimal Solution: The Hold Minimum (HM) Algorithm}
The HM algorithm solves the problem \emph{optimally}. We explain the algorithm
by an example. Consider the power cost matrix which is given by
\begin{equation}
    \mathbf{P} = \begin{bmatrix}
                   9 & 3 \\
                   1 & 4 \\
                   2 & 8
                 \end{bmatrix}.
\end{equation}
The power cost matrix can also be shown as $\mathbf{P} =
(p_1,p_2,\ldots,p_n)$, where $p_j = (p_{1j},p_{2j},\ldots,p_{mj})^T$. In each
step, the algorithm removes a group of values $p_{ij}$ of the power cost
matrix. Removing $p_{ij}$ means that we eliminate those clusters that include
the mobile $i$ and BS $j$ from the collection set. The algorithm compares
maximum $n$ clusterings and holds only the clustering minimizing the total cost. Thus, it terminates in a step $Q$ where each mobile is assigned to only one BS. In step $s$, the power cost matrix and collection set is denoted as
$\mathbf{P}[s]=(p_1[s], p_2[s], \ldots, p_n[s])$ and $\mathcal{S}[s]$,
respectively.

Let us now turn to the example. In the initial step $s = 0$, we assume that
$\mathbf{P}[0] = \mathbf{P}$ and $\mathcal{S}[0] = \mathcal{S}$ given by
\begin{gather}
\mathcal{S}[0] = \{(1;1), (2;1), (3;1), (1,2;1), (1,3;1), (2,3;1), \notag\\
(1,2,3;1), (1;2), (2;2), (3;2), (1,2;2), (1,3;2), \notag\\
(2,3;2), (1,2,3;2)\}.
\end{gather}
Recall that assigning a cluster of mobiles $S_l$ to BS $j$ has a cost $\max_{i\in S_l} p_{ij}$. Therefore, if we find the maximum value of $p_j$, we obtain the total cost in case of all mobiles in column $j$ are assigned to BS $j$. For example, $\max p_1 = \max (9,1,2) = 9$. This means that if all mobiles are assigned to BS 1, then the total cost is $9$.

The algorithm runs as following: in step $s=1$, we find the maximum value of
each column of power cost matrix, then eliminate all values in power cost
matrix except minimum of the calculated maximum values. Namely, $\max(9,1,2) =
9$ and $\max(3,4,8) = 8$, then $9$ is eliminated by putting an $\infty$
\begin{equation}
    \mathbf{P}[1] = \begin{bmatrix}
                   \infty & 3 \\
                   1 & 4 \\
                   2 & 8
                 \end{bmatrix}.
\end{equation}
Thus, the collection set reduces to the following
\begin{gather}
\mathcal{S}[1] = \{(2;1), (3;1), (2,3;1), (1;2), (2;2), (3;2), (1,2;2), \notag
\\ (1,3;2), (2,3;2), (1,2,3;2)\}.
\end{gather}
First column contains an $\infty$ which means that mobile $1$ must be assigned
to another BS (i.e. in this example, obviously BS $2$). In fact, this
represents the recursiveness of the algorithm where we run the algorithm for a
sub power cost matrix. In this simple example, the sub power cost matrix is
$3$. In general case, the algorithm does the following
\begin{equation}
    P_j[s] = \left\{
                  \begin{array}{ll}
                    \max p_{j}[s]  + \mathfrak{h}(\mathbf{P}^{sub}_{j}[s]), &
                    \hbox{if sub power cost matrix;} \\
                    \max p_{j}[s], & \hbox{otherwise.}
                  \end{array}
                \right.
\end{equation}
where $P_j[s]$ represents the total cost if we assign all mobiles to BS $j$
except those that can not be assigned to, and the optimal cost occurring
due to
sub power cost matrix $\mathbf{P}^{sub}_j[s]$ in step $s$. Here, $\mathfrak{h}
: \Re^{m\times n} \rightarrow \{\Re, \mathcal{A}\}$ is the function
which gives the optimal value and assignments obtained by running HM algorithm.

For $s=2$, we calculate $P_1[2] = \max(1,2) + \mathfrak{h}(3) = 2 + 3=5$,
where $\mathbf{P}^{sub}_1[2] = (3)$. On the other hand, we do not need to
calculate $P_2[2]$ since it is kept in the memory. Therefore, $P_2[2] =
P_2[1]$. Then, the algorithm holds minimum value of $\min(P_1[2],P_2[2]) =
P_1[2] = 5$ meaning that we remove $8$ resulting in the following
\begin{equation}
    \mathbf{P}[2] = \begin{bmatrix}
                   \infty & 3 \\
                   1 & 4 \\
                   2 & \infty
                 \end{bmatrix},
\end{equation}
and
\begin{gather}
\mathcal{S}[2] = \{(2;1), (3;1), (2,3;1), (1;2), (2;2), (1,2;2)\}.
\end{gather}
Then, for $s=3$, $P_2[3] = \max(3,4) + \mathfrak{h}(2) = 4+2=6$, where
$\mathbf{P}^{sub}_2[3] = (2)$, and $P_1[3]=P_1[2]$. We remove $4$, since
$\min(P_1[3],P_2[3]) = P_1[3]$. This gives the following matrix, collection
set, assignments, and optimal total power,
\begin{equation}
    \mathbf{P}[3] = \begin{bmatrix}
                   \infty & 3 \\
                   1 & \infty \\
                   2 & \infty
                 \end{bmatrix},
\end{equation}
\begin{gather}
\mathcal{S}[3] = \{(2;1), (3;1), (2,3;1), (1;2)\} = \{(2,3;1), (1;2)\},
\end{gather}
\begin{equation}
    \mathfrak{h}(\mathbf{P}) := \left\{ p = 5, \mathbf{A} = \begin{bmatrix}
                   0 & 1 \\
                   1 & 0 \\
                   1 & 0
                 \end{bmatrix}
                 \right\},
\end{equation}
respectively. The pseudo-code of this algorithm is given in Algorithm \ref{alg:holdmin}.

\begin{theorem}
\emph{The HM algorithm terminates in finite step. At this step, the total power cost is minimum.}
\end{theorem}
\begin{proof}
The power cost matrix is transformed in each step by
\begin{equation}
\mathbf{P}[0]\rightarrow \mathbf{P}[1] \rightarrow \ldots \rightarrow \mathbf{P}[Q].
\end{equation}
In each step, at least $n-1$ values are removed from the power cost matrix.
Removed values are those that increase the total cost. The value that is held
is the minimum one in the corresponding step. So, in the terminal step $s =
Q$, it arrives to such a cost that is the lowest. Each mobile is assigned to
exactly one BS in the terminal step.
\end{proof}

\begin{algorithm}
\caption{The Hold Minimum}
\label{alg:holdmin}
\algorithmicfunction $(p,\mathbf{A})=\mathfrak{h}(\mathbf{P})$
\begin{algorithmic}
\WHILE{each row includes several ``$1$''s in $\mathbf{A}$}
    \STATE $k \leftarrow 1$
    \FOR{$j = $ indices of columns of $\mathbf{P}$ not including all $\infty$}
        \STATE $V_k \leftarrow$ maximum of column $j$ except $\infty$
        \IF{any row in column $j$ includes $\infty$}
            \STATE $\mathbf{P}^{sub} \leftarrow$ rows including $\infty$ of $\mathbf{P}$
            \STATE $p_{sub} \leftarrow \mathfrak{h}(\mathbf{P}^{sub})$
            \STATE $V_k \leftarrow V_k + p_{sub}$
        \ENDIF
        \STATE $k \leftarrow k + 1$
    \ENDFOR
    \STATE $(i_{min}, j_{min}, V_{min}) \leftarrow \min{V}$
    \STATE Hold $V_{min}$ using $(i_{min}, j_{min})$ and put $\infty$ in all indices causing $P_{ij} \geq V_{min}$ in $\mathbf{P}$
    \STATE Put $0$ in all indices causing $P_{ij} \geq V_{min}$ in $\mathbf{A}$
\ENDWHILE
\STATE $p \leftarrow \sum{\max{\mathbf{A}\bigotimes\mathbf{P}}}$
\end{algorithmic}
\algorithmicendfunction
\end{algorithm}

The complexity of this algorithm is assessed in Section \ref{sec:Complexity}.

\subsection{Greedy Solution: The Column Control (CC) Algorithm}
Let us denote the CC algorithm by $\mathfrak{c}:\Re^{m\times n}
\rightarrow \{\Re, \mathcal{A}\}$. This algorithm exploits the assumption that the operational power cost $P_0$ is fairly higher than the transmitted power, i.e. $P_0 \gg P_{ij}$. The aim here is to assign many mobiles to only one BS. Recall that cluster $N_j$ (set of mobiles that can be assigned to BS $j$) has the transmission cost $\max_{i\in N_j} P_{ij}$ and the operational cost $P_0^j$. For better understanding, consider the following power cost matrix in which the operational power cost is assumed to be $12$ $W$:
\begin{equation}\label{powerMatrixCC}
    \mathbf{P} = \begin{bmatrix}
                   12.50 & 12.40 & 12.32 & \infty \\
                   12.30 & 12.30 & 12.43 & \infty \\
                   12.20 & 12.45 & 12.15 & 12.23 \\
                   \infty & 12.43 & 12.25 & 12.35 \\
                   \infty & \infty & \infty & 12.29
                 \end{bmatrix}.
\end{equation}
In such a scenario, we can assign $|N_1| = 3$ mobiles to BS 1 with cost $12.50$, $|N_2| = 4$ mobiles to BS 2 with cost $12.45$, $|N_3| = 4$ mobiles to BS 3 with cost $12.43$, and $|N_4| = 3$ mobiles to BS 3 with cost $12.35$.

The logic behind the CC algorithm is the following:
\begin{enumerate}
  \item \emph{Find how many mobiles can be assigned to each BS}
  \item \emph{Choose the BS to which it can be assigned the most mobiles}
  \item \emph{If there are multiple BSs in the state of 2), then choose the BS which can serve the mobiles with minimal cost}
\end{enumerate}
Applying these rules to the last example, it turns out that BS 3 can cover the most mobiles $|N_3| = 4$ which are $N_3 = (1,2,3,4)$ with the lowest cost $12.43$. Then, the algorithm assigns only a cluster of mobiles. In the following step, the CC algorithm performs the same rules to the
remained mobiles which produces the sub power cost matrix $\mathbf{P}^{sub}$.
In the last example, it is given by
$
    \mathbf{P}^{sub} = \begin{bmatrix}
                   \infty & \infty & \infty & 12.29
                 \end{bmatrix}.
$
This results in the assignment of mobile 5 to BS 4, because mobile 5 can not be assigned to other BSs. This shows the recursiveness of the CC algorithm. A pseudo-code is given in Algorithm \ref{alg:columncontrol}.

Formally, in step $s$, we denote as following the set of mobiles assigned to
BS $j$:
$
R[s] = \left\{ \textrm{Assigned mobiles in step $s$} \right\}
$. Thus, the collection set is reduced as following in step $s$:
\begin{equation}
\mathcal{S}[s] = \left\{ \mathcal{S}[s-1] \setminus (S;k) : i\in S, \forall i
\in R[s] \textrm{ and } \forall k\in N\setminus j\right\}.
\end{equation}
Considering the last example, in step $s=1$, $R[1] = (1,2,3,4)$. By assigning
these mobiles to BS 3, the CC algorithm removes those assignments which
include mobiles $(1,2,3,4)$ except $(S;k) = \{(1,2,3,4),3\}$. Thus, the total
power and assignments are given by
\begin{equation}
\mathfrak{c}(\mathbf{P}) := \left\{p = 24.72,
\mathbf{A} = \begin{bmatrix}
                   0 & 0 & 1 & 0 \\
                   0 & 0 & 1 & 0 \\
                   0 & 0 & 1 & 0 \\
                   0 & 0 & 1 & 0 \\
                   0 & 0 & 0 & 1
                 \end{bmatrix}
\right\}.
\end{equation}

Moreover, for any power cost matrix, we conclude that in a final step $Q$, the CC algorithm converges to the case where each mobile is assigned to one BS.

\begin{algorithm}
\caption{The Column Control}
\label{alg:columncontrol}
\algorithmicfunction $(p,\mathbf{A}) = \mathfrak{c}(\mathbf{P})$
\begin{algorithmic}
\STATE $v' = $ Find how many $\infty$ has each column of $\mathbf{P}$
\STATE $v_{min} \leftarrow \min{v'}$
\STATE $v'' = $ Find which row of $v'$ is equal to $v_{min}$
\FOR{$l = $ columns of $\mathbf{P}$ determined by $|v''|$}
    \STATE $V_l \leftarrow$ maximum value of column $l$ of $\mathbf{P}$
\ENDFOR
\STATE Find $V_{min} = \min{V_l}$ and the corresponding column $l_{min}$
\IF {$|v'|==0$}
    \STATE Put 1s to the column $l_{min}$ in $\mathbf{A}$
\ELSE
    \STATE $\mathbf{A}' \leftarrow$ Put 1s to the column $l_{min}$ in $\mathbf{A}$
    \STATE Find sub power cost matrix $\mathbf{P}^{sub}$ which is composed by
    those mobiles that are not assigned
    \STATE $\mathbf{A}^{sub} \leftarrow$ Find assignments by running $\mathfrak{c}(\mathbf{P}^{sub})$
\ENDIF
\STATE $\mathbf{A} \leftarrow$ Combine $\mathbf{A}^{sub}$ and $\mathbf{A}'$
\STATE $p \leftarrow \sum{\max{\mathbf{A}\bigotimes\mathbf{P}}}$
\end{algorithmic}
\algorithmicendfunction
\end{algorithm}

\subsection{Distributed Column Control (DCC) Algorithm}\label{dcc}
Assume that each BS broadcasts its own power vector and identities of mobiles
that it can serve. Recall that we denote the power vector of BS $j$ as $p_j$.
The power vector of BS 1 given in the power cost matrix of eq.
(\ref{powerMatrixCC}) is $p_1 = (12.50, 12.30, 12.20, \infty, \infty)^T$. From a practical point of view, the BS broadcasts only the identity and associated costs of the mobiles which are associated with it. The infinity values are not broadcast. Therefore, the BS broadcasts the power vector as $p_1 = (12.50, 12.30, 12.20)^T$, and an identity vector represented as $h_j = (h_{jk})\in \mathbb{N}^{|p_j|}$. For example, $h_1 = (1,2,3)^T$.

Moreover, each mobile receives power vectors from all BSs that can transmit to
it. Then, each mobile generates the power cost matrix from received power
vectors. For example, mobile 1 receives from BS 1, BS 2, and BS 3 the power
vectors \[p_1 = (12.50, 12.30, 12.20)^T,\] \[p_2 = (12.40, 12.30, 12.45,
12.43)^T,\] and \[p_3 = (12.32, 12.43, 12.15, 12.25)^T\] with identity vectors
$h_1 = (1,2,3)^T$, $h_2 = (1,2,3,4)^T$, and $h_3 = (1,2,3,4)^T$, respectively.
Mobile 1 decides that the power cost matrix is as following:
\begin{equation}
                \begin{bmatrix}
                   12.50 & 12.40 & 12.32 \\
                   12.30 & 12.30 & 12.43 \\
                   12.20 & 12.45 & 12.15 \\
                   \infty & 12.43 & 12.25 \\
                 \end{bmatrix}.
\end{equation}
Note that mobile 1 realizes from $h_2$ and $h_3$ that BS 1 can not transmit to
mobile 4. Therefore, it puts an infinite cost corresponding to mobile 4 and BS
1 in the power cost matrix.

By this rule each mobile determines its own power cost matrix. Thus, each
mobile finds the assignments according to CC algorithm, and selects the BS
from which it will receive data. For example, mobile 1 obtains the following
assignment matrix by running CC algorithm
\begin{equation}
                \begin{bmatrix}
                   0 & 0 & 1 \\
                   0 & 0 & 1 \\
                   0 & 0 & 1 \\
                   0 & 0 & 1 \\
                 \end{bmatrix}.
\end{equation}
It turns out that mobile 1 chooses BS 3 for reception the broadcast data.

\begin{remark}
Through the advantage of the decentralization, the DCC algorithm makes possible the following: \emph{if the mobiles do not send any assignment information to a BS, then the corresponding BS is considered to be switched off}. On the other hand, the CC algorithm is centralized, therefore, \emph{the network determines according to the assignments which BS is switched off}.
\end{remark}

\subsection{Greedy Solution: The Nearest BS (NBS) Algorithm}
To solve the problem heuristically, the easiest way is to assign each mobile to the nearest BS. By ``nearness'', we do not mean a geographical measure, instead, it is the lowest power cost that the corresponding mobile needs from the corresponding BS. So, the mobile selects the BS transmitting with the lowest power. We assume that a mobile is capable to know the power costs corresponding to those BSs that can transmit to it.

Clearly, mobile $i$ knows the vector $p_i = (p_{i1}, p_{i2},\ldots, p_{in_i})$ where $n_i$ denotes the number of BSs that mobile $i$ can be served. Then, mobile $i$ only calculates the minimal value of $p_i$, and chooses the corresponding BS,
\begin{equation}
a_{i,j} = 1: \quad j = \arg\min_{j} p_i,
\end{equation}
where $a_{i,j} \in \mathbf{A}$ is the BS that mobile $i$ selects. Actually, this corresponds to remove from collection set $\mathcal{S}$ all assignments related to mobile $i$ and the BSs being out BS $j$.

The NBS algorithm is very efficient and quick, and in most cases, it gives optimal assignments when the operational power cost $P_0$ is neglected (See Section \ref{sec:simulationresults}).

\subsection{Greedy Set-Cover Algorithm (greedy-SC)}
In \cite{chvatal}, a greedy heuristic for the set covering problem is proposed by Chvatal. We utilize this algorithm in simulation results (section \ref{sec:simulationresults}) to compare it with the proposed algorithms.

%The pseudo-code of the nearest BS algorithm is given in Algorithm 1.
%
%\begin{algorithm}\label{nearestbs}
%\caption{The Nearest BS}
%\small
%\algorithmicfunction nearestbs($\mathbf{P}$)
%\begin{algorithmic}
%    \STATE $A_{i,j} \leftarrow 0, \forall i,j$
%    \FOR{$i = $ indices of rows of $\mathbf{P}$}
%        \STATE $(V_i,I_i) \leftarrow \min(P_i)$
%        \STATE $j \leftarrow I_i$
%        \STATE $A_{i,j} \leftarrow 1$
%    \ENDFOR
%    \STATE $p \leftarrow \sum{\max{\mathbf{A}\bigotimes\mathbf{P}}}$
%\end{algorithmic}
%\algorithmicendfunction
%\end{algorithm}

\section{Time Complexity Analysis}
\label{sec:timecomplexity}
In this section, we estimate and compare the time complexity of the proposed algorithms.

%Let us assume that $m = k n$. The input size is supposed to be the total
%number of elements of the power cost matrix, denoted as $x = nm$. This choice
%provides us to calculate the time complexity $T(x)$ in terms of $x$. It is
%straightforward to obtain that $n = \sqrt{\frac{x}{k}}$ and $m = \sqrt{kx}$.
\begin{theorem}
\emph{The time complexity of the HM algorithm is $O(m^3 n) + (n-1)\sum_{s=1}^{m} T(m-s,n-1) + m O(n)$}.
\end{theorem}

\begin{proof}

%\begin{equation}
%T(m,n) = \sum_{i=1}^m T_s(i) h_m(i) + \sum_{j = 1}^n T_s(j) h_n(j)
%\end{equation}
%
%\begin{equation}
%T(m,n) = \sum_{i\in {\cal H}_m} \ell_i + \sum_{j\in {\cal H}_n} \ell_j
%\end{equation}

Note that at most $n-1$ operations are required to find
the maximum value of a column of the power cost matrix in one step, and a single
operation for finding the minimum of a vector having $n$ values. In each step, there is at least one element in a column that is removed. At most
$m$ steps are needed until a convergence to the optimal case. Selection
algorithms (finding maximum or minimum) have a time complexity in $O(n)$
\cite{complexityofmaximum}. We can express as following the time complexity including the recursive property of the algorithm:
\begin{equation}
T(m,n) = \sum_{s=1}^m \left[ (n-1)\left( O(m-s) + T(m-s,n-1) \right) +O(n) \right]
\end{equation}
where we assume that in each step $ s $ a sub-matrix occurs with complexity $ T(m-s,n) $. Solving this kind of recurrence is tricky but further technical manipulations lead to following expression:
\begin{equation}
T(m,n) = O(m^3 n) + (n-1)\sum_{s=1}^{m} T(m-s,n-1) + m O(n).
\end{equation}
We also simulate the last recurrence resulting in the following graphical characteristics:
\begin{equation*}
\includegraphics[width=\linewidth]{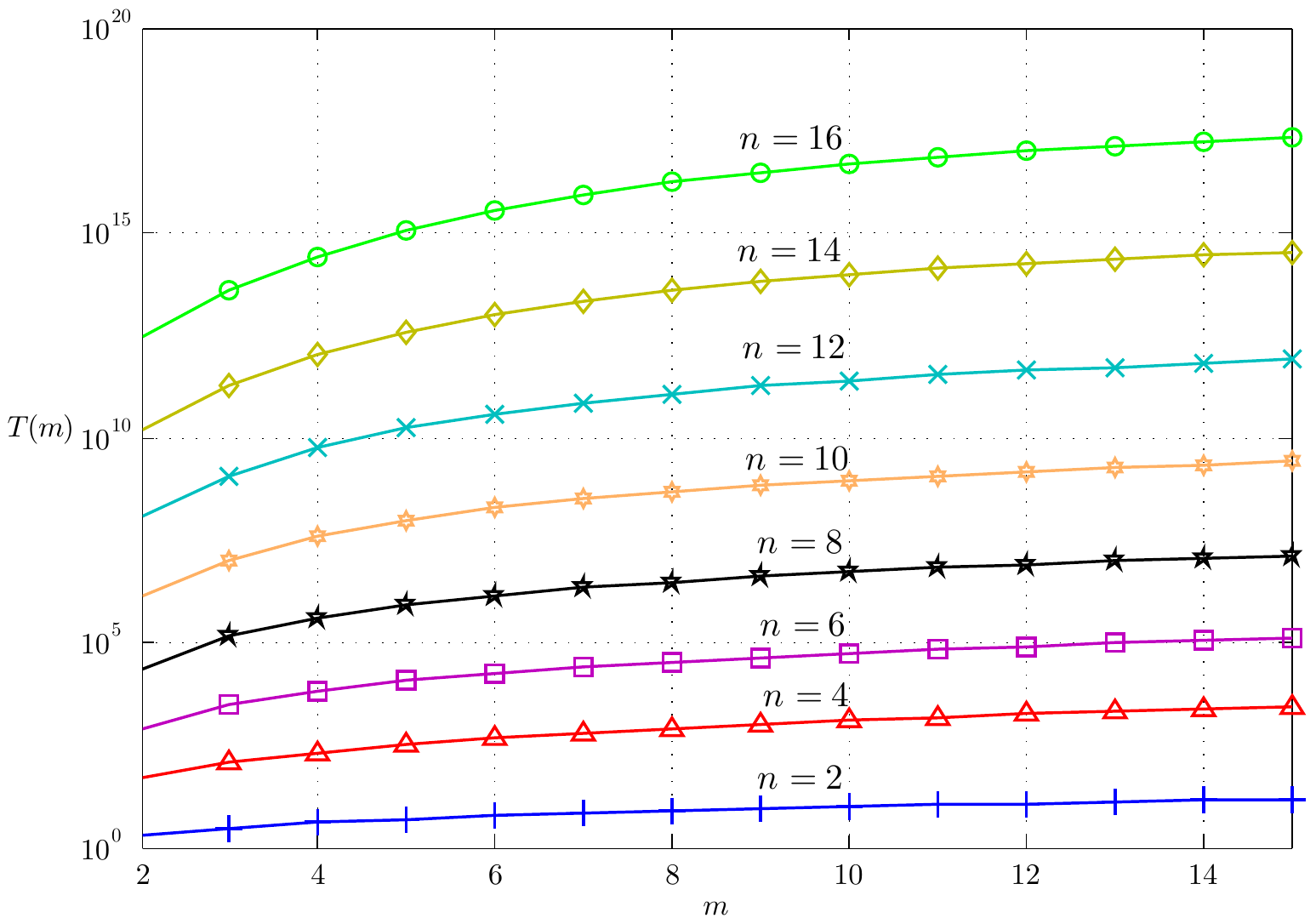}
\end{equation*}
In the axis denoting $ T(m) $, a log-scale is used. The figure illustrates that the complexity grows exponentially.
\end{proof}

\begin{theorem}
\emph{The time complexity of the CC algorithm is $O\left( mn
\right)$}.
\end{theorem}
\begin{proof}
This result is known and a complete proof is available \cite{bib:ApproximateClustering}.
\end{proof}

\begin{theorem}
\emph{The time complexity of the NBS algorithm is $O\left(n \right)$}.
\end{theorem}
\begin{proof}
It is straightforward because the only operation that is performed in the NBS algorithm is to find the minimum value of a vector of power cost matrix having a dimension at most $n$. Therefore, the time complexity is $O\left( n \right)$.
\end{proof}

\begin{theorem}
The time complexity of the HD algorithm is PLS-complete with $ \rho(m,n)mO\left(n\right) $.
\end{theorem}
\begin{proof}
We denote by $ \rho(m,n) $ the rounds needed for converging to the Nash equilibrium. Note that this is a random variable. In each step, only one mobile performs a maximization operation which causes totally $ m O(n) $ complexity. This is repeated during $ \rho(m,n) $ rounds. So, the time required can be given by $ T(m,n) = \rho(m,n)mO\left(n\right) $. Actually, the complexity of finding the Nash equilibria in a potential game is known to be PLS-complete, refer to \cite{bib:ThecomplexityofpureNashequilibria}.

Furthermore, the computational results achieved in next section highlight  that \textit{the number of rounds grows logarithmically with respect to the number of mobiles and BSs} (Figure \ref{fig:meanNumOfRoundsVSArea}).
\end{proof}

%
%\subsection{The complexity of NBS}
%
%\begin{equation}
%    T(x) = x^{1/2}(T_{m}(x) + 1 + 1) + x + T_{m}(x) + x^{1/2}
%\end{equation}
%The worst-case time complexity of $\max()$ and $\max()$ is equal each other which is given by $T_{m}(x)=x^{1/2} + \log_2(x^{1/2})$ [ref].
%\begin{equation}
%    T(x) = 2x + 4x^{1/2} + (x^{1/2}+1)\log_2(x^{1/2})
%\end{equation}
%$O(x)$
%
%\subsection{The complexity of HM}
%
%\begin{equation}
%    T(x) = x^{1/2}(T_{m}(x) + 1 + 1) + x + T_{m}(x) + x^{1/2}
%\end{equation}
%The worst-case time complexity of $\max()$ and $\max()$ is equal each other which is given by $T_{m}(x)=x^{1/2} + \log_2(x^{1/2})$ [ref].
%\begin{equation}
%    T(x) = 2x + 4x^{1/2} + (x^{1/2}+1)\log_2(x^{1/2})
%\end{equation}
%$O(x)$

\section{Simulation Results}
\label{sec:simulationresults}
\subsection{Downlink System Model}
For small-cells, the cellular network model consists of BSs arranged according to an homogeneous Poisson point process $\Phi$ of intensity $\lambda_b$ $(points/m^2)$ in the Euclidean plane \cite{bacelli}. For macro-cells, we use the classical honeycomb model to represent a well structured network made of large cells. 

Also, we consider an independent collection of mobile users, located according to some independent homogeneous Poisson point process with intensity $\lambda_m$ $(points/m^2)$. The main weakness of the Poisson model is that because of the independence of the Poisson point process, two BSs have a non null probability to be located very close to each other. This weakness is balanced by two strengths: the natural inclusion of different cell sizes and shapes and the lack of edge effects, i.e. the network extends indefinitely in all directions \cite{andrewsbacelli}. The expected value of a homogeneous Poisson point process is ${\bf E}[\Phi] = \lambda A$, where $A \subset \Re^2$ denotes some area. 
%The snapshot depicted in Figure \ref{fig:distOfBSsAndMobiles} shows the distribution of BSs and mobiles with intensity $\lambda_b = 44.4\times 10^{-5}$ $\frac{points}{m^2}$ and $\lambda_m = 2.5\times 10^{-5}$ $\frac{points}{m^2}$ where $A = 9$ $km^2$.

Moreover, the deployment scenario used to generate Figures \ref{fig:pAvVsLambdaBNBS}, \ref{fig:pAvVsLambdaBwithC0}, \ref{fig:meanPowerVSLambdamSmallCell}, \ref{fig:meanPoverMaxThetaPVSThetaSmallCell} corresponds to small-cells. We assume $P_r = -80$ $dBm$, being the typical maximum received signal power of a wireless network as well as we set arbitrarily $P_{ij} = \infty$ if $P_{ij} \geq 20$ $dBm$ and we set the path loss exponent $\alpha = 3$. We also assume an equal operational power cost for all BS, $P_0 = 12$ $W$. 
%Furthermore, we enumerate the BSs and mobiles according to their distance to the origin arbitrarily chosen (Figure \ref{fig:distOfBSsAndMobiles}).

\subsection{Performance Results}
We compare the proposed algorithms for different values of $\lambda_m$ and $\lambda_b$. %We also aimed to obtain the results in a range milliseconds of running time in order to calculate the average total power when comparing the  algorithms with the HM algorithm. 
The average total power was calculated by Monte Carlo simulations by running
the algorithms for different generated power cost matrices% for some iteration number 
and taking the mean of the results.
%\begin{align}
%%\begin{split}
%    & P(0) = 0, \notag \\
%    & P(i+1) = P(i) + p, \quad i={0,\ldots,t} \notag\\
%    & \bar{p} = \frac{P(t+1)}{t}.
%%\end{split}
%\end{align}

We first start with macro-cells. As mentioned above, we use a honeycomb model with a cell radius equal to $4 km$. Table \ref{tab:comperisonOfHMNBS} presents the optimal set-covering (SC) result as well as those obtained with the different proposed algorithms for different realization of power cost matrices and when the operational power cost is null, i.e. $P_0 = 0$ for all BSs. The transmission power $P_{ij} = \infty$ if $P_{ij} \geq 48$ $dBm$. It turns out that the HD algorithm is very efficient  and converges to nearly optimal assignments when operational power costs are neglected. This result indicates that switching off some BS does not decrease the total power significantly. The NBS algorithm also produces near optimal results in many examples. In these practical scenarios, the idea of switching off some BS is mostly interesting when the circuit power is dominant, which is not the case for macro-cells. Further, as fast moving users are associated with macro-cells in priority, it seems reasonable to keep all active. Therefore, the NBS algorithm is efficient and the gain achievable with any other optimal algorithm is marginal. 

Table \ref{tab:comperisonOfHMCCDCCNBS} compares the SC optimal results with all developed algorithms introduced in the paper for small-cells scenarios. The operational power cost is set to $P_0 = 12$ $W$ and the BS density is increased. CC and DCC algorithms produce optimal assignments for almost all examples. However, the DCC algorithm naturally performs worse when the number of BSs increases. Moreover, the NBS algorithm exhibits worst results which highlights the interest of switching off some BSs when the operational power cost is higher than the transmission power cost. %On the other hand, the HD algorithm is also efficient when the parameter $\theta$ is calibrated properly. In many scenarios, it gives optimal or near-optimal results. In the sequel, we explain the advantages of the HD algorithm compared to the others. 

Figure \ref{fig:pAvVsLambdaBNBS} illustrates the results achieved with the NBS algorithm, for different densities of BS and mobiles. This figure highlights an intuitive property of the NBS approach. When no operational cost is considered (see Figure \ref{fig:pAvVsLambdaBNBS}), the power consumption decreases with the BS density, since the average distance between mobiles and BSs decreases accordingly. On the opposite, when an operational cost is considered, the NBS algorithm leads to an increased energy consumption since the number of active BS increases accordingly.

Let us now switch to the CC and DCC algorithms. Figure \ref{fig:pAvVsLambdaBwithC0} focuses only on scenarios with circuit power, because they correspond to the more relevant cases. the upper curves show that the CC algorithm achieves much better results than NBS, especially when the BS density is high. This algorithm privileges solutions with larger cells. The distributed version DCC performs worst than the centralised one but still better than the NBS.

Figure \ref{fig:meanPowerVSLambdamSmallCell} plots the change of the average total power with respect to the intensity of mobiles for small-cells scenario. The assumptions are as following: $\lambda_b = 1.11 \times 10^{-5} \frac{points}{m^2}$, $\theta = 0.002$ (in Figure \ref{fig:meanPoverMaxThetaPVSThetaSmallCell}, the optimal $ \theta $ is found), and area $A = 6.25 km ^2$. Note that the HD algorithm performs efficiently even though it is decentralized. For example, in case of $ \lambda_m = 8\times 10^{-8} $, the average number of mobiles is given by $8\times 10^{-8}\cdot 6.25\times 10^{6} = 50$; thus, the average power used per mobile is calculated as following: a) the HD algorithm: $ 63/50 = 1.26 W $, b) the CC algorithm: $ 55/50 = 1.1 W $, c) the greedy-SC algorithm: $ 50/50= 1 W$, d) the SC algorithm: $ 43.83/50 = 0.88 W $.
 
Figure \ref{fig:meanPoverMaxThetaPVSThetaMacroCell} depicts the change of the average total power with respect to the intensity of mobiles for macro-cell deployment. $ \lambda_b = \frac{80 points}{3600 km^2} $, $ A = 3600 km^2 $ ($ 60 km \times 60 km $ area), and $ \theta = 0.21 $ (in Figure \ref{fig:meanPoverMaxThetaPVSThetaMacroCell}, we plot the change of average total power with respect to $ \theta $, and choose the optimal value). Here, we observe that the HD algorithm produces remarkable results. Calibrating $ \theta $ properly is significant, otherwise the HD algorithm may not converge to the near optimal results. On the other hand, the NBS algorithm is also efficient in the macro-cell deployment. The drawback of greedy-SC algorithm reveals here since it works with a mechanism where the larger cells are privileged.

In Figures \ref{fig:meanPoverMaxThetaPVSThetaSmallCell} and \ref{fig:meanPoverMaxThetaPVSThetaMacroCell}, the normalized average total power is plotted with respect to $ \theta $. From the figures and our observations in experiments performed in MATLAB, it might be considered that $ \theta $ is mainly affected by the area over which the algorithm runs. For example, in Figure \ref{fig:meanPoverMaxThetaPVSThetaMacroCell}, the normalized average total power has a minimum in the same value of intensity of BSs, but it moves to a higher value when the area is enlarged from $ 2500 km^2 $ to $ 3600 km^2 $.

Figure \ref{fig:meanNumOfRoundsVSArea} shows the change of the average number of rounds of the HD algorithm for converging to a Nash equilibrium with respect to the area. The figure implies that the average number of rounds has a logarithmic characteristic. Moreover, when the operational power costs are zero, the average number of rounds increases since smaller cells are formed; therefore, the HD algorithm needs more rounds to converge to a Nash equilibrium. 

\begin{table*}
\caption{$A=2500 km^2$, $\lambda_b = \frac{6 points}{2500 km^2}$,
$\lambda_m = \frac{1 point}{25 km^2}$, $P_0=0$ $W$, $\theta = 0.11$.}
\centering
\label{tab:comperisonOfHMNBS}
\scriptsize
\begin{tabular}{| c || c | c | c | c | c | c | c | c | c |}
\hline
Ex. $i$ & $m$ & $n$ &  & HM & SC & CC & DCC & NBS & HD\\
\hline
\hline
1 & 54 & 6 & $10^{-7}\times$ & 294.0360 & 294.0360 & 294.0360 & 295.5705 & 294.0360 & 294.0360\\
\hline
2 & 43 & 6 & $10^{-7}\times$ & 299.8159 & 299.8159 & 323.6194 & 352.1657 & 299.8159 & 299.8159\\
\hline
3 & 46 & 6 & $10^{-7}\times$ & 250.4830 & 250.4830 & 271.4828 & 271.4828 & 250.4830 & 250.4830\\
\hline
4 & 41 & 6 & $10^{-7}\times$ & 270.4417 & 270.4417 & 302.8145 & 287.9684 & 284.5738 & 283.0740\\
\hline
5 & 51 & 6 & $10^{-7}\times$ & 307.7226 & 307.7226 & 361.7673 & 361.7673 & 320.2662 & 307.7226\\
\hline
6 & 45 & 6 & $10^{-7}\times$ & 278.5206 & 278.5206 & 317.5086 & 317.5086 & 278.5206 & 278.5206\\
\hline
7 & 41 & 6 & $10^{-7}\times$ & 305.1243 & 305.1243 & 345.3096 & 345.3096 & 306.8924 & 312.2107\\
\hline
8 & 34 & 6 & $10^{-7}\times$ & 221.5681 & 221.5681 & 236.7168 & 236.7168 & 256.0736 & 221.5681\\
\hline
9 & 58 & 6 & $10^{-7}\times$ & 360.3562 & 360.3562 & 363.1885 & 363.1885 & 360.3562 & 360.3562\\
\hline
10 & 52 & 6 & $10^{-7}\times$ & 310.0721 & 310.0721 & 310.0721 & 310.0721 & 332.0735 & 332.0735\\
\hline
11 & 44 & 6 & $10^{-7}\times$ & 283.9244 & 283.9244 & 313.8116 & 339.6551 & 283.9244 & 283.9244\\
\hline
12 & 49 & 6 & $10^{-7}\times$ & 312.9718 & 312.9718 & 312.9718 & 312.9718 & 325.5796 & 337.9936\\
\hline
13 & 53 & 6 & $10^{-7}\times$ & 229.9336 & 229.9336 & 255.3125 & 255.3125 & 229.9336 & 232.6530\\
\hline
14 & 60 & 6 & $10^{-7}\times$ & 308.6002 & 308.6002 & 308.6002 & 308.6002 & 320.2059 & 320.2059\\
\hline
15 & 57 & 6 & $10^{-7}\times$ & 329.2667 & 329.2667 & 342.8460 & 354.0650 & 329.2667 & 335.9528\\
\hline
\end{tabular}
\end{table*}
\begin{table}
\caption{$A=4$ $km^2$, $\lambda_b = \frac{6 points}{4 km^2}$,
$\lambda_m = \frac{18 points}{4 km^2}$, $P_0=12$ $W$, $\theta = 0.003$.}
\centering
\label{tab:comperisonOfHMCCDCCNBS}
\scriptsize
\begin{tabular}{| c || c | c | c | c | c | c | c | c |}
\hline
Ex. $i$ & $m$ & $n$ & HM & SC & CC & DCC & NBS & HD\\
\hline
\hline
1  &  14 & 4 & 24.153 & 24.153 & 24.153 & 24.153 & 48.183 & 24.153\\
\hline
2  & 14 & 7 & 24.140 & 24.140 & 36.226 & 36.226 & 72.139 & 24.152\\
\hline
3  &  23 & 7 & 36.138 & 36.138 & 48.177 & 48.275 & 84.042 & 48.185\\
\hline
4  &  16 & 7 & 36.195 & 36.195 & 36.195 & 36.257 & 72.208 & 48.204\\
\hline
5  & 19 & 4 & 24.109 & 24.109 & 24.109 & 36.174 & 48.126 & 36.184\\
\hline
6  &  7 & 3 & 12.070 & 12.070 & 12.070 & 12.070 & 24.073 & 12.070\\
\hline
7  &  15 & 8 & 36.178 & 36.178 & 36.178 & 36.222 & 84.241 & 36.222\\
\hline
8  &  15 & 7 & 24.110  & 24.110 & 24.110 & 24.144 & 72.101 & 24.154\\
\hline
9  & 18 & 9 & 36.099 & 36.099 & 36.148 & 36.253 & 84.042 & 36.148\\
\hline
10 &  24 & 6 & 36.111 & 36.111 & 36.191 & 36.191 & 60.134 & 36.191\\
\hline
11 & 21 & 3 & 24.149 & 24.149 & 24.149 & 24.149 & 36.162 & 24.149\\
\hline
12 & 18 & 10 & 36.111 & 36.111 & 36.111 & 48.141 & 84.053 & 48.152\\
\hline
13 & 23 & 5 & 36.105 & 36.105 & 36.105 & 48.143 & 60.218 & 36.105\\
\hline
14 & 13 & 4 & 24.082 & 24.082 & 24.104 & 24.104 & 48.123 & 24.110\\
\hline
15 & 17 & 8 & 36.190 & 36.190 & 36.190 & 36.217 & 60.217 & 36.190\\
\hline
\end{tabular}
\end{table}

\begin{table}
\caption{$\lambda_b = 0.10\times 10^{-3}\frac{points}{m^2}$,
$\lambda_m = 1.11\times 10^{-3} \frac{points}{m^2}$, $P_0=12$ $W$, $\theta = 0.008$.}
\centering
\label{tab:numOfRounds}
\begin{tabular}{| c || c | c | c | c |}
\hline
\multicolumn{1}{|c||}{} & \multicolumn{4}{c|}{Number of rounds}\\
\hline
Example/Area($km^2$) & $0.98$ & $1.28$ & $1.62$ & $2.00$\\
\hline
\hline
1  &  3 & 2 & 3 & 3 \\
\hline
2  & 3 & 3 & 3 & 3 \\
\hline
3  &  3 & 3 & 2 & 3 \\
\hline
4  &  2 & 3 & 3 & 3 \\
\hline
5  & 3 & 2 & 3 & 3 \\
\hline
6  &  3 & 3 & 3 & 3 \\
\hline
7  &  3 & 3 & 3 & 4 \\
\hline
8  &  3 & 3 & 3  & 4 \\
\hline
9  & 3 & 3 & 3 & 3 \\
\hline
10 &  3 & 3 & 3 & 4 \\
\hline
\end{tabular}
\end{table}

\begin{figure}
    \centering    \includegraphics[width=\linewidth]{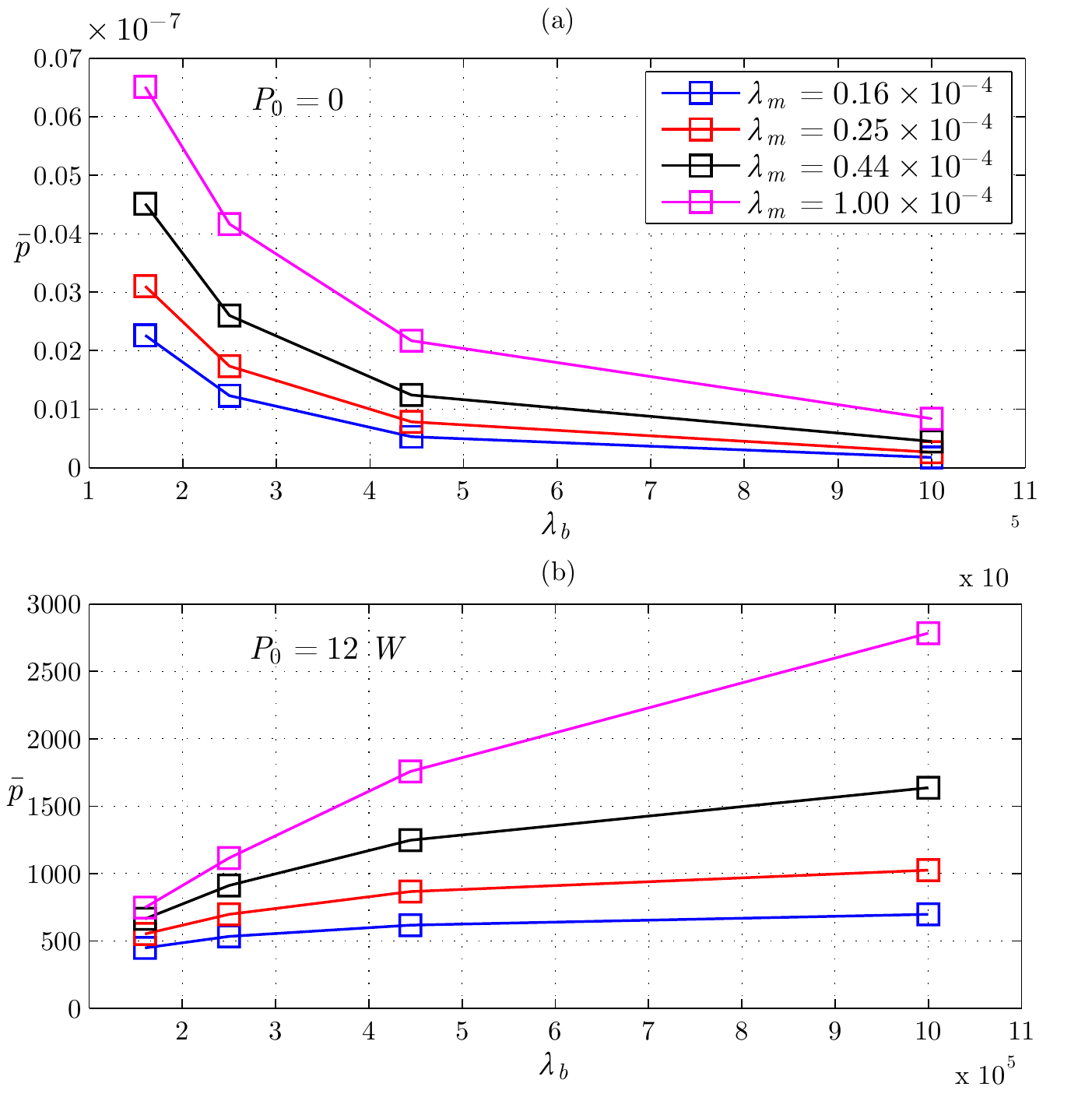}
\caption{The NBS Algorithm: Average total power $\bar{p}$ with respect to intensity of BSs $\lambda_b$ for increasing values of intensity of mobiles $\lambda_m = (0.16, 0.25, 0.44, 1.00)\times 10^{-4}$ $\frac{points}{m^2}$, $A = 4$ $km^2$.}\label{fig:pAvVsLambdaBNBS}
\end{figure}
\begin{figure}
    \centering    \includegraphics[width=\linewidth]{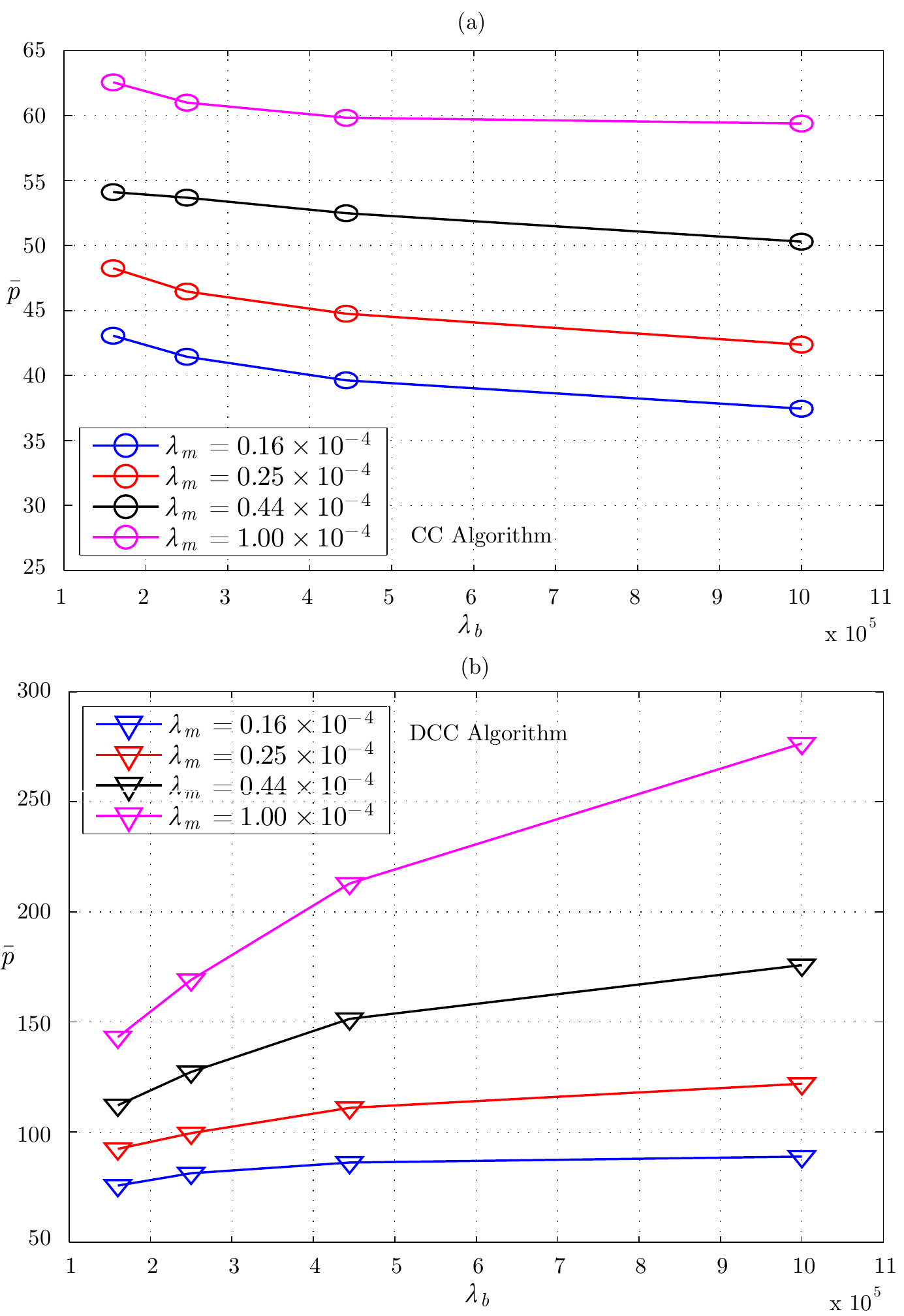}
\caption{Average total power $\bar{p}$ with respect to intensity of BSs $\lambda_b$ for increasing values of intensity of mobiles $\lambda_m = (0.16, 0.25, 0.44, 1.00)\times 10^{-4}$ $\frac{points}{m^2}$, $A = 4$ $km^2$, 
$P_0=12$ $W$.}\label{fig:pAvVsLambdaBwithC0}
\end{figure}
%\begin{figure}
%    \centering    \includegraphics[width=\linewidth]{graphics/pAvVsLambdaBwithC0_DCC.pdf}
%\caption{The DCC Algorithm: Change of the average total power $\bar{p}$ with 
%respect to intensity of BSs $\lambda_b$ for increasing values of intensity of 
%mobiles $\lambda_m = (0.16, 0.25, 0.44, 1.00)\times 10^{-4}$ 
%$\frac{points}{m^2}$, $A = 4$ $km^2$, 
%$P_0=12$ $W$.}\label{fig:pAvVsLambdaBwithC0DCC}
%\end{figure}
\begin{figure}
    \centering    \includegraphics[width=\linewidth]{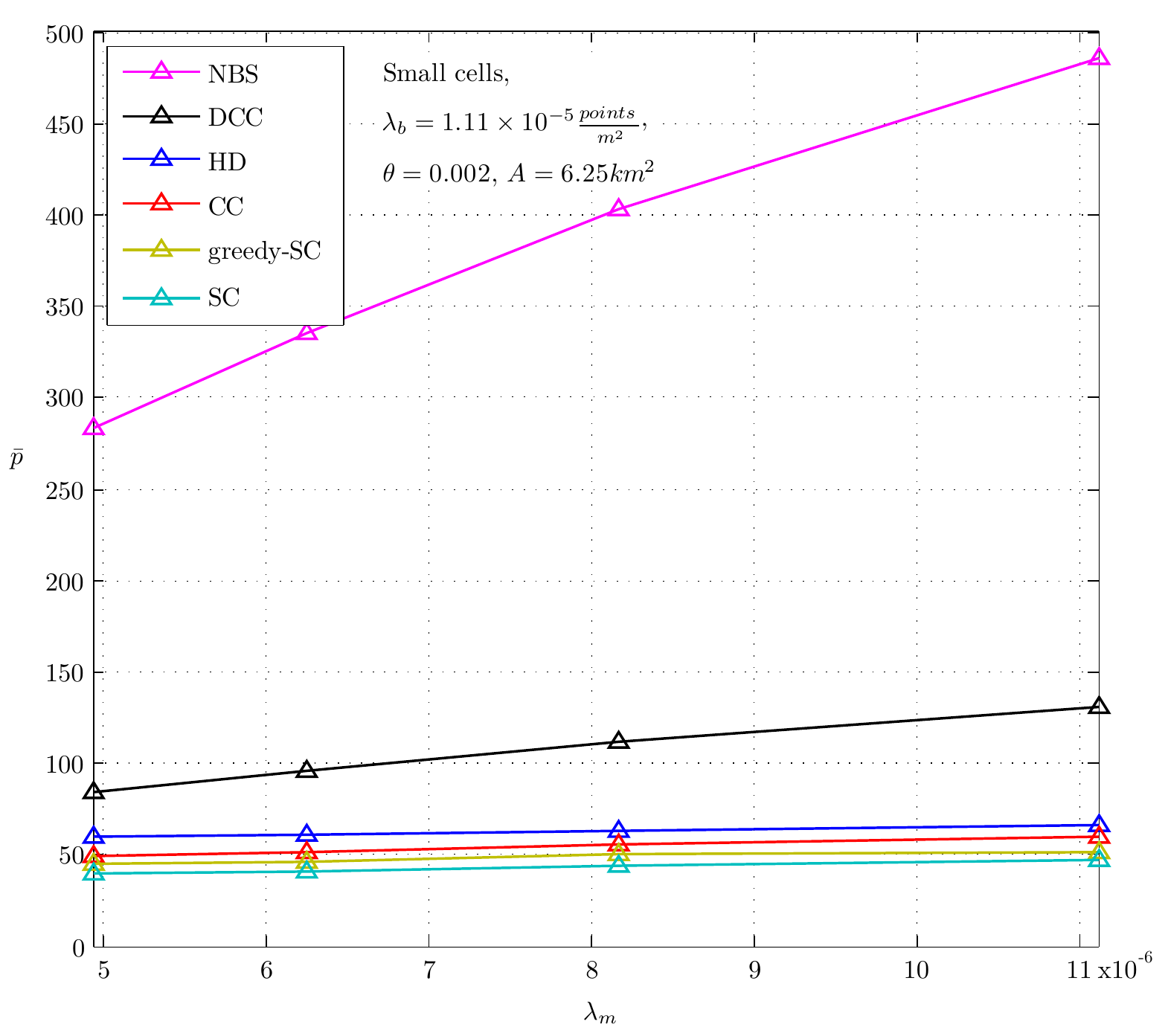}
\caption{Small cells: Change of the average total power $\bar{p}$ with 
respect to the intensity of mobiles.}
\label{fig:meanPowerVSLambdamSmallCell}
\end{figure}
\begin{figure}
    \centering    \includegraphics[width=\linewidth]{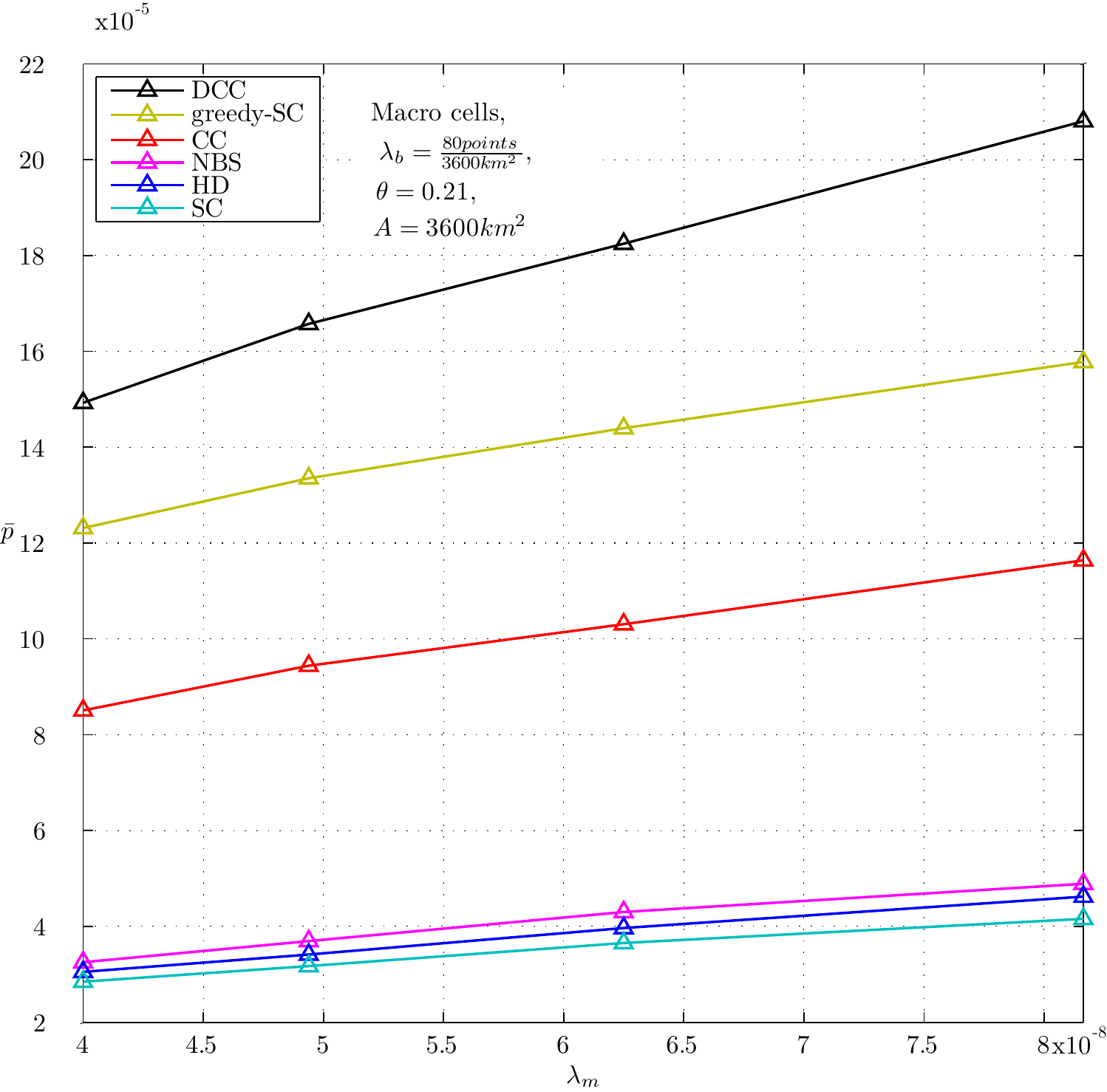}
\caption{Macro cells: Average total power $\bar{p}$ with 
respect to the intensity of mobiles.}
\label{fig:meanPowerVSLambdamMacroCell}
\end{figure}
\begin{figure}
    \centering    \includegraphics[width=\linewidth]{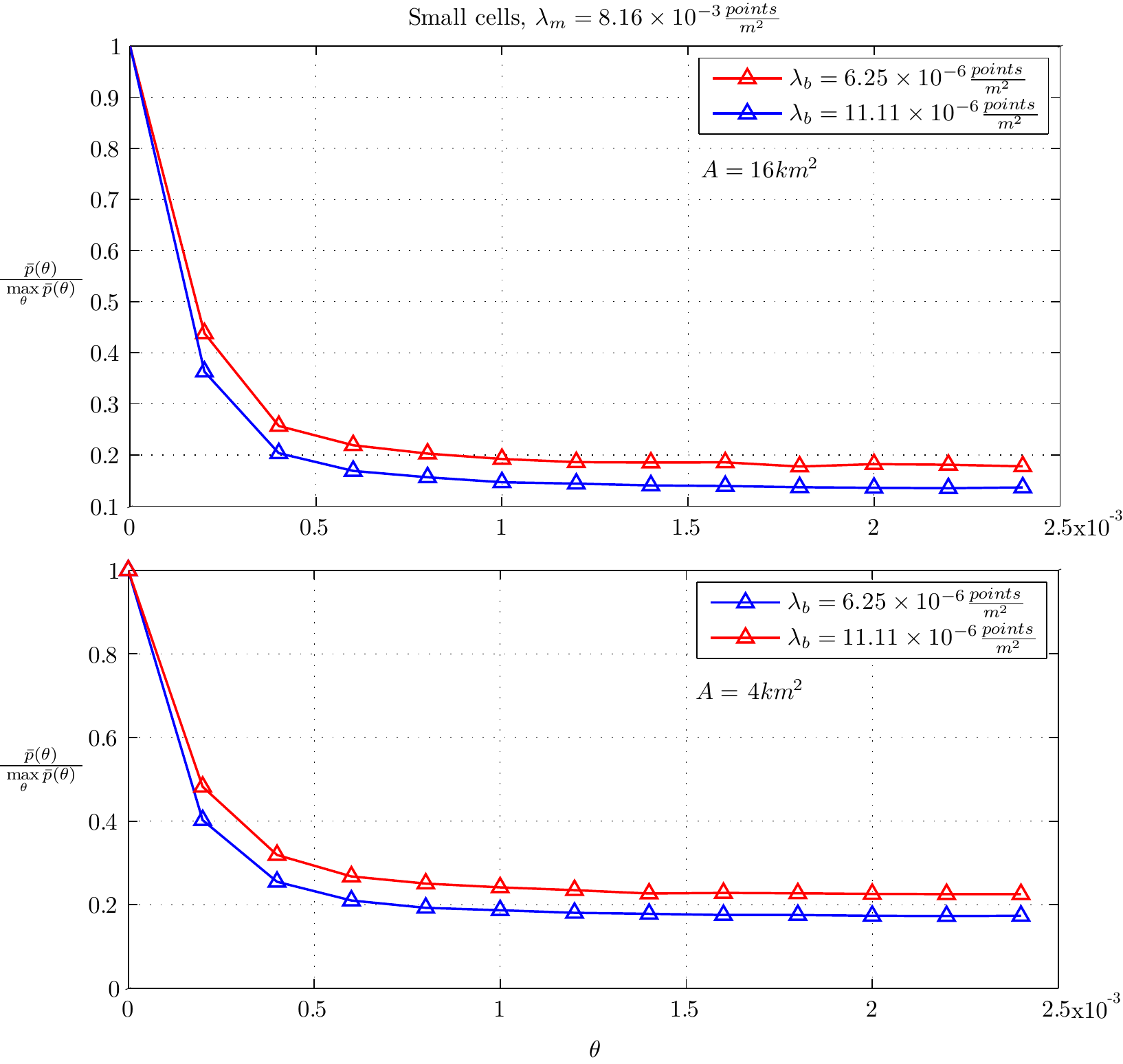}
\caption{Small cells: Normalized average total power with respect to $\theta$.}
\label{fig:meanPoverMaxThetaPVSThetaSmallCell}
\end{figure}
\begin{figure}
    \centering    \includegraphics[width=\linewidth]{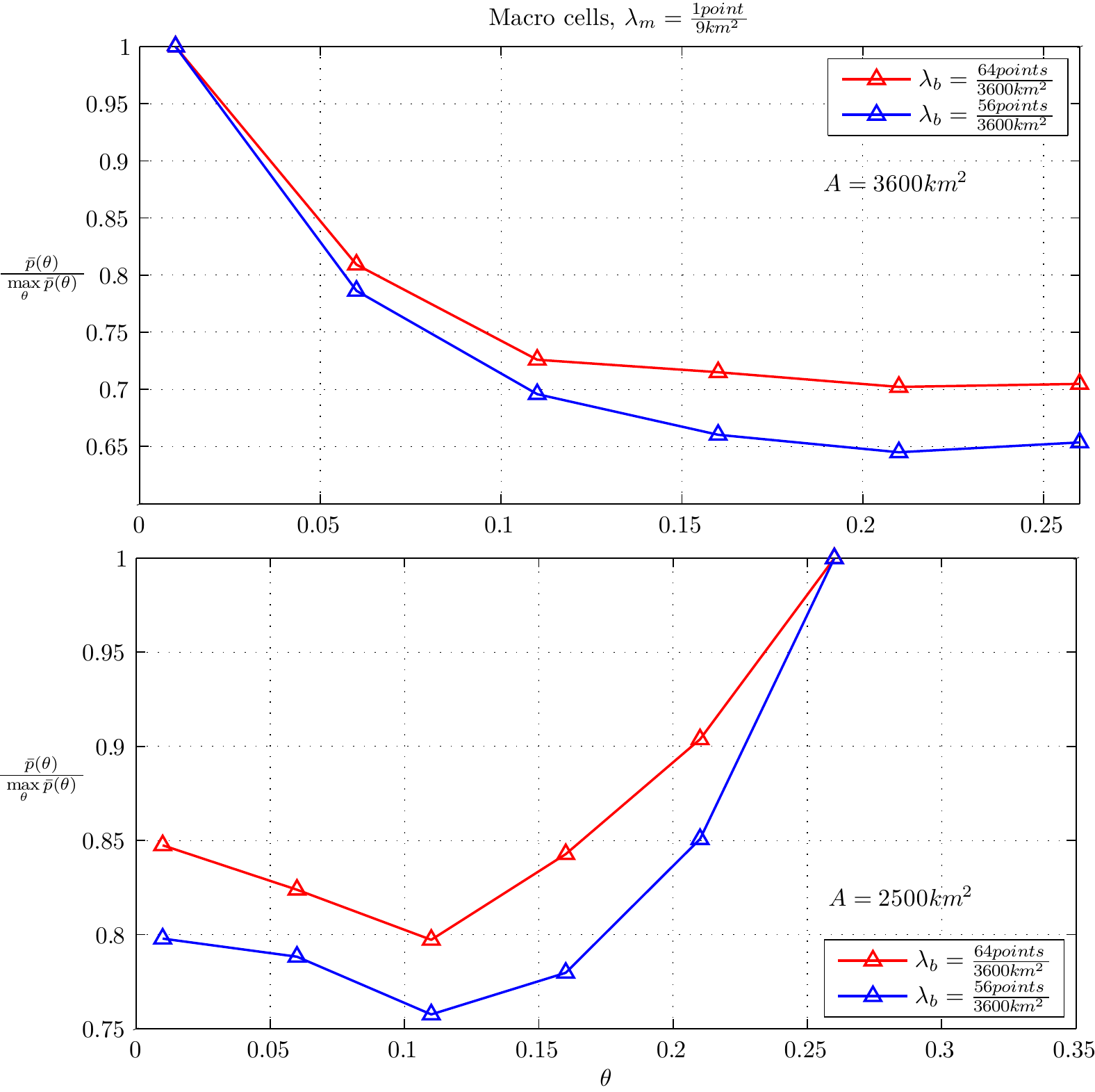}
\caption{Macro cells: Normalized average total power with respect to $\theta$.}
\label{fig:meanPoverMaxThetaPVSThetaMacroCell}
\end{figure}

\begin{figure}
    \centering    \includegraphics[width=\linewidth]{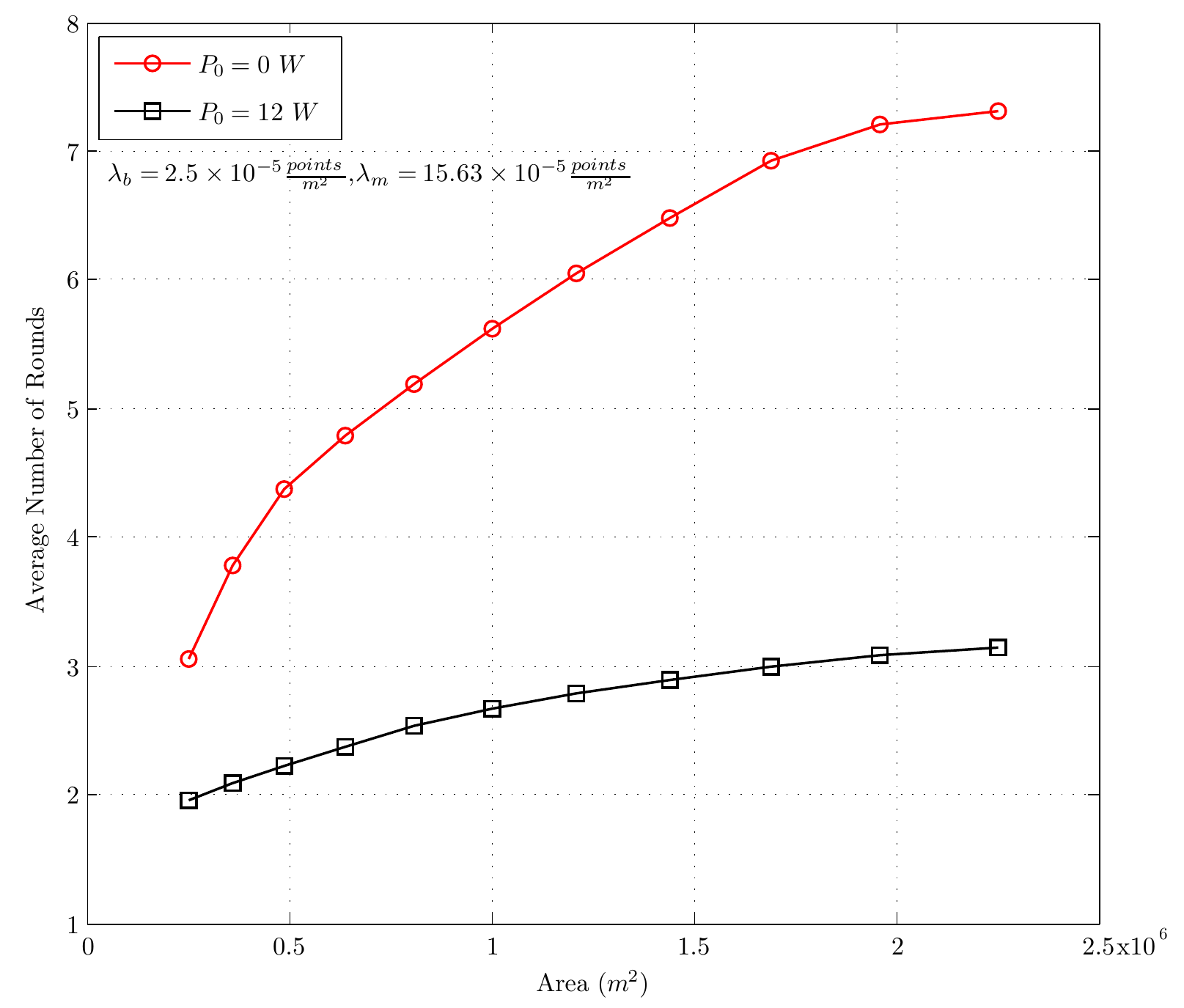}
\caption{Average number of rounds with respect to the area.}
\label{fig:meanNumOfRoundsVSArea}
\end{figure}
%\begin{figure}[!hbp]
%    \centering    \includegraphics[width=\linewidth]{ratioOfNBSandDCC.eps}
%\caption{The ratio of the average total power of the NBS algorithm and the DCC 
%algorithm with respect to intensity of BSs $\lambda_b$, 
%$P_0=12$ $W$.}\label{fig:ratioOfNBSandDCC}
%\end{figure}

\section{Conclusion and Future Works}
\label{sec:conclusions}
This paper addressed the MAP problem in the context of broadcast transmission. We introduced a novel decentralized solution based on group formation games, which we named the hedonic decision (HD) algorithm. This formalism is constructive: a new class of group formation games is introduced where the utility of players within a group is separable and symmetric being a generalization of party affiliation games. We proposed a centralized optimal recursive algorithm (the HM) as well as a centralized polynomial-time heuristic algorithm (the CC). The results exhibit that the HD algorithm achieves very good results if the parameter $ \theta $ is well chosen. The exact value of $ \theta $ is not provided and may be used as a setting parameter. 

The proposed HD algorithm is efficient and may be used in many other set covering problems. For instance, indoor wireless network planning has been studied for several years and the optimal BS activation is an important problem. The proposed HD algorithm could be used  for planning purposes and thus optimizing the number of BS, but could also be used to optimize dynamically the number of active BS as a function of actives users. The provider may deploy a high density of BS and then could run dynamically the HD algorithm to optimize then number of active BSs.

Furthermore, it could be interesting to change the game parameters. One can define different clustering weights for each BS. That setting could provide better results. It is also possible to try different symmetric bipartite utility allocations.

\end{document}